\documentclass[11pt,a4paper]{article}
\usepackage[a4paper]{geometry}
\usepackage{amssymb,latexsym,amsmath,amsfonts,amsthm}
\usepackage{amsmath}
\usepackage{amsthm}
\usepackage{amssymb}
\usepackage{amsfonts}
\usepackage{epsf}
\usepackage{graphicx}
\usepackage{epstopdf}
\usepackage{hyperref}
\usepackage{color}
\usepackage{cancel}
\usepackage{comment}
\usepackage{fullpage}
\def\Im {\mathop{\rm Im}\nolimits}
\def\arg {\mathop{\rm arg}\nolimits}
\def\Re {\mathop{\rm Re}\nolimits}
\def\Ai {{\rm Ai}}

\newtheorem{thm}{Theorem}[section]

\newtheorem{lem}[thm]{Lemma}
\newtheorem{prop}[thm]{Proposition}

\newtheorem{asm}[thm]{Assumption}

\numberwithin{equation}{section}

\theoremstyle{remark}
\newtheorem{rem}[thm]{Remark}
\numberwithin{equation}{section}

\newcounter{comment}
\numberwithin{equation}{section}

\title{{Gaussian unitary ensembles with pole singularities near the soft edge and a system of coupled Painlev\'{e} XXXIV equations}}

\author{Dan Dai\footnotemark[1], ~Shuai-Xia Xu\footnotemark[2] ~and Lun Zhang\footnotemark[3]}

\footnotetext[1]{Department of Mathematics, City University of Hong Kong, Tat Chee
Avenue, Kowloon, Hong Kong. E-mail: \texttt{dandai@cityu.edu.hk}}
\footnotetext[2]{Institut Franco-Chinois de l'Energie Nucl\'{e}aire, Sun Yat-sen University,
Guangzhou 510275, China. E-mail: \texttt{xushx3@mail.sysu.edu.cn} (Corresponding author)}
\footnotetext[3] {School of Mathematical Sciences and Shanghai Key Laboratory for Contemporary Applied Mathematics, Fudan University, Shanghai 200433, China. E-mail: \texttt{lunzhang@fudan.edu.cn }}

\date{}

\begin{document}

\maketitle

\noindent \hrule width 6.27in\vskip .3cm

\noindent {\bf{Abstract }}
In this paper, we study the singularly perturbed Gaussian unitary ensembles defined by the measure
\begin{equation*}
\frac{1}{C_n}  e^{- n\textrm{tr}\, V(M;\lambda,\vec{t}\;)}dM,
\end{equation*}
over the space of $n \times n$ Hermitian matrices $M$, where $V(x;\lambda,\vec{t}\;):= 2x^2 + \sum_{k=1}^{2m}t_k(x-\lambda)^{-k}$ with
$\vec{t}= (t_1, t_2, \ldots, t_{2m})\in \mathbb{R}^{2m-1} \times (0,\infty)$,
in the multiple scaling limit where $\lambda\to 1$ 
together with $\vec{t} \to \vec{0}$ as $n\to \infty$ at appropriate related rates.
We obtain the asymptotics of the partition function, which is described explicitly in terms of an integral involving a smooth solution to a new coupled Painlev\'e system generalizing the Painlev\'e XXXIV equation. The large $n$ limit of the correlation kernel is also derived, which leads to a new universal class built out of the $\Psi$-function associated with the coupled Painlev\'e system.


\vspace{.2in} \noindent {\it {Keywords: }}random matrices; singularly perturbed Gaussian unitary ensembles;
Riemann-Hilbert approach; asymptotics of the partition function; limiting correlation kernel; Painlev\'{e} type equations
\vskip .3cm

\noindent \hrule width 6.27in\vskip 1.3cm

\newpage

\tableofcontents

\section{Introduction}
In this paper, we are concerned with the following singularly perturbed Gaussian unitary ensembles (GUEs)
\begin{equation}\label{pGUE}
\frac{1}{C_n}  e^{- n\textrm{tr}\, V(M;\lambda,\vec{t}\;)}dM,
\end{equation}
defined on the space $\mathcal{H}_n$ of $n \times n$ Hermitian matrices $M=(M_{ij})_{1\leq i,j \leq n}$, where
\begin{align}
dM & =\prod_{i=1}^{n}dM_{ii} \prod_{1 \leq i < j \leq n} d\Re M_{ij}d\Im M_{ij},
\\
C_n & =C_n(\lambda; \vec{t}\;)  =\int_{\mathcal{H}_n}  e^{- n\textrm{tr} V(M;\lambda,\vec{t}\;)}dM
\label{partition function}
\end{align}
is the normalization constant,  and the potential
\begin{equation} \label{potential:V}
  V(x;\lambda,\vec{t}\;):= 2x^2 + \sum_{k=1}^{2m}t_k(x-\lambda)^{-k}, \quad \vec{t}\;= (t_1, t_2, \ldots, t_{2m})\in \mathbb{R}^{2m-1}\times (0,\infty),
\end{equation}
%with $\vec{t}\;= (t_1, t_2, \ldots, t_{2m})\in \mathbb{R}^{2m-1}\times (0,\infty)$, $\lambda\in \mathbb{R}$ and $m\in \mathbb{N}$.
with $\lambda\in \mathbb{R}$ and $m\in \mathbb{N}$.

Since the ensembles are unitary invariant, we have (cf. \cite{deift,mehta}) that the $n$ eigenvalues $x_1,\ldots,x_n$ of $M$ from \eqref{pGUE} induce the following
probability density function
\begin{equation}\label{eq:jpdf}
\frac{1}{Z_n(\lambda)}\prod_{1\leq i < j \leq n}(x_j-x_i)^2\prod_{j=1}^n w(x_j),
\end{equation}
where
\begin{equation}\label{eq:weight}
w(x)=w(x;\lambda,\vec{t}\;)=e^{-n V(x;\lambda,\vec{t}\;)}
\end{equation}
and
\begin{equation}\label{def:partiationfunction}
Z_n(\lambda) = Z_n(\lambda;\vec{t}\;)=\int_{\mathbb{R} ^n}\prod_{1\leq i < j \leq n}(x_j-x_i)^2\prod_{j=1}^n w(x_j)dx_j
\end{equation}
is the partition function. It is also easily seen that the distribution \eqref{eq:jpdf} is determinantal with respect to a correlation kernel $K_{n}(x,y;\lambda,\vec{t}\;)$ which can be constructed out of the orthogonal polynomials associated with the weight function \eqref{eq:weight} over $\mathbb{R}$. Indeed, let $\pi_j(x)=\pi_{j}(x;\lambda,\vec{t}\;)$, $j=0,1,\ldots,$ be the family of monic polynomials of degree $j$ satisfying
\begin{equation}\label{eq:orthogonal}
\int_{\mathbb{R}}\pi_j(x)\pi_m(x)w(x)dx=\gamma_j(\lambda;\vec{t}\;)^{-2}\delta_{j,m}.
\end{equation}
Then, the correlation kernel can be written as
\begin{equation} \label{eq:correlation kernel}
K_{n}(x,y;\lambda,\vec{t}\;)=\gamma_{n-1}^{2} \sqrt{w(x) w(y)}
\frac {\pi_n(x)\pi_{n-1}(y)-\pi_{n-1}(x)\pi_n(y)}{x-y}.
\end{equation}

Obviously, if $\vec{t} = \vec{0}$ or $\lambda \to \infty$, the model \eqref{pGUE} reduces to the classical GUE.
A well-known fact is that the limiting eigenvalue distribution of GUE, or equivalently, the macroscopic limit of the correlation kernel is described by the Wigner's semicircle law whose density function is given by
\begin{equation} \label{semicircle-law}
  \rho_{\textrm{sc}}(x)= \frac{2}{\pi} \sqrt{1-x^2}, \qquad x \in [-1,1].
\end{equation}
The local statistics of the eigenvalues obeys the principle of universality. This means that, after proper centering and scaling, the large $n$ limit of the correlation kernel tends to the sine kernel for $x\in(-1,1)$ (bulk universality), and to the Airy kernel for $x=\pm 1$ (soft edge universality). If the vector % $\vec{t} \neq \vec{0}$ 
$\vec{t}\;= (t_1, t_2, \ldots, t_{2m})\in \mathbb{R}^{2m-1}\times (0,\infty)$ is fixed, however, due to the presence of the pole singularity located at $x=\lambda$, the eigenvalues are pushed away from $\lambda$ and it is unlikely to find the eigenvalues near the pole as the matrix size $n$ becomes large. It would then be interesting to consider the case that $\vec{t} \to \vec{0}$ and $n \to \infty$ simultaneously at appropriate related rates. Although the limiting mean distribution of the eigenvalues remains unchanged in this regime, which is still given by the semicircle law \eqref{semicircle-law}, it is expected that some new phenomena
will occur near $x=\lambda$, which can be interpreted as a description of the phase transition between different edge behaviors.

Apart from the theoretical interest stated above, the studies of invariant random matrix models with singular potentials are also justified due to their frequent occurrences in mathematical physics, and significant progresses have been achieved over the past few years. Partially motivated by the distribution of zeros of the Riemann zeta function on the critical line (cf. Berry and Shukla \cite{BS08}), Mezzadri and Mo \cite{Mez:Mo}, Brightmore et al. \cite{Bri:Mezz:Mo} considered the following perturbed GUE, defined by the measure
\begin{equation}\label{model-Bri}
\frac{1}{\widehat{C}_n}e^{-n \textrm{tr}\, \left(\frac{1}{2}M^2-\frac{t}{M}+\frac{z^2}{2M^2}\right)}dM,
\end{equation}
over $\mathcal{H}_n$. Clearly, this corresponds to $\lambda =0$ and $m=1$ in \eqref{pGUE}. In the double scaling regime that both $z$ and $t$ are of order $O(n^{-1/2})$, a phase transition in the $(z,t)$-plane characterized by the Painlev\'e III equation was discovered in \cite{Bri:Mezz:Mo}. In the context related to an integrable quantum field theory at finite temperature, Chen and Its \cite{ci} considered a perturbed Laguerre unitary ensemble over the space $\mathcal{H}^+_n$ of $n \times n$ positive definite Hermitian matrices whose potential possesses a simple pole at the origin, which is defined by the measure
\begin{equation} \label{model-ChenIts}
  \frac{1}{\widetilde {C}_n} (\det M)^\alpha e^{- \mathrm{tr}\, \left(M+\frac{t}{M}\right)} dM, \quad \alpha > -1, \quad t>0.
\end{equation}
They studied the moment generating function when the matrix size $n$ is fixed. 
When the parameter $t=0$, the ensemble \eqref{model-Bri} is closely related to  \eqref{model-ChenIts} after a change of variables and the statistics in  \eqref{model-Bri} can be derived by using  the statistics in \eqref{model-ChenIts}  with $\alpha = \pm1/2$. The asymptotic studies of this model were later carried out by Xu et al. in \cite{XDZ2015,XDZ}. It comes out that in the double scaling limit where $t \to 0^+$ as $n \to \infty$, the hard edge scaling limit of the correlation kernel and the asymtotics of the partition function are all related to the Painlev\'e III equation. Particularly, the new limiting kernel provides a description of the transition between the classical Airy kernel and
the Bessel kernel. The results in \cite{XDZ2015,XDZ} were further extended by Atkin et al. in \cite{ACM}, where they considered the case that a fairly general class of potentials perturbed by a pole of order $k\in \mathbb{N}$. A hierarchy of higher order analogues to the Painlev\'e III equation was used to describe the double scaling limits of the partition function and the correlation kernel; see also our recent work \cite{Dai:Xu:Zhang2018} on the properties of the Fredholm determinant associated with this family of limiting kernels (known as the gap probability). Other problems related to the singularly perturbed ensembles include the field of spin glasses \cite{Akemann14},  eigenvalues of Wigner-Smith time-delay matrix in the context of quantum transport and electrical characteristics of chaotic cavities \cite{Bro:Fra:Bee,Mez:Sim,Tex:Maj},
and the bosonic replica field theories \cite{Osip:Kanz2007}.

It is worthwhile to point out the role played by the location of the pole. The eigenvalue distribution in a ``merging" regime and in an ``evaporating" regime have already been reported by Akemann et al. in \cite{Akemann14}, depending on whether the pole is located inside the bulk of the limiting spectrum or not. The known critical behavior of the eigenvalues near the pole, as just reviewed, corresponds to the choice that the pole is located inside the bulk or at the hard edge of the limiting spectrum. In both cases, the Painlev\'e III equation and its hierarchy are essential in describing the critical behaviors. It is then natural to raise the following question:
\begin{itemize}
  \item What is the local behavior of the eigenvalues near the soft edge if the pole approaches the soft edge as well?
\end{itemize}

In the present work, we intend to answer this question by establishing a multiple scaling limit of the correlation kernel for the perturbed GUEs \eqref{pGUE} in the sense that $\lambda$ approaches the soft edge together with $\vec{t} \to \vec{0}$ as $n\to \infty$. Moreover, we also obtain the asymptotics of the partition function \eqref{def:partiationfunction} under the same regime.  Instead of the Painlev\'e III equation (or its hierarchy), our results will be described by a new system of nonlinear ODEs generalizing the Painlev\'e XXXIV equation, as stated in what follows.

%--------------------------------------------------------------------------------------------------------------------------------------------------------------------------------
\section{Statement of results}
\label{sec:results}

\subsection*{A new coupled Painlev\'{e} XXXIV  system}
The asymptotics of the partition function involves a special solution to a coupled Painlev\'{e} system,
%we need t The model RH problem for $\Psi$ is related to the following system of $2m+1$ ODEs indexed by $p=2m+2,2m+3,\ldots,2m+4$,
which is defined by  $2m+1$ ODEs indexed by $p=2m+2,2m+3,\ldots,4m+2$,
\begin{equation}\label{eq:b-k}
\sum_{k=p-2m-1}^{2m+1}(b_{p-k}b_k''-\frac {1}{2}b_k'b_{p-k}'-2(2b_1+s)b_{p-k}b_k-2b_{p-k}b_{k+1}) +  2 \tilde{\tau}_{p}=0,
\end{equation}
for $2m+1$ unknown functions $b_1=b_1(s),\ldots,b_{2m+1}=b_{2m+1}(s)$, where $\tilde{\tau}_{p}$ are real constants and
$$b_{k}=0, \qquad k>2m+1.$$

Note that if $m=0$, the system \eqref{eq:b-k} reduces to a single ODE
\begin{equation}\label{eq:p34}
b_1''=4b_1^2+2sb_1+\frac {b_1'^2-4\tilde{\tau}_2}{2b_1},
\end{equation}
which is known as the Painlev\'{e} XXXIV equation; cf. \cite{InceBook}.
The differential system \eqref{eq:b-k} can then be regarded as a generalization of the Painlev\'{e} XXXIV equation, and we call it a coupled Painlev\'{e} XXXIV system.

Our first result concerns the existence of a special solution to the above coupled Painlev\'{e} system.
\begin{thm}\label{thm:painleve solution}
Let $(\tau_1, \tau_2,\ldots,\tau_{2m})\in \mathbb{R}^{2m-1}\times (0, +\infty)$ be any fixed vector. Then, there exists pole-free solutions $b_1(s),\ldots,b_{2m+1}(s)$ to the coupled Painlev\'{e} XXXIV system \eqref{eq:b-k} for real values of $s$ with the parameters $\tilde \tau_p$ given by
\begin{equation}\label{def:tauptilde}
\tilde{\tau}_p=\sum_{k=p-2m-1}^{2m+1}(k-1)(p-k-1)\tau_{k-1}\tau_{p-k-1}, \qquad 2m+2\leq p \leq 4m+2.
\end{equation}
Moreover, as $s \to +\infty$, we have that
\begin{equation}\label{eq:b1asy}
b_1(s)=-\frac{\tau_1}{2s^{\frac32}}+O(s^{-5/2}).
\end{equation}
\end{thm}

\begin{rem}
In the literature (cf. \cite{cjp}), the Painlev\'{e} XXXIV hierarchy is defined by
\begin{equation}\label{eq:P34 hierarchy}
(2\mathcal{L}_n[U]-s)\frac{d^2}{ds^2}(\mathcal{L}_n[U])-\left(\frac{d}{ds}(\mathcal{L}_n[U])\right)^2+\frac{d}{ds}(\mathcal{L}_n[U])+(2\mathcal{L}_n[U]-s)^2U+\alpha_n=0, 
\end{equation}
where $\alpha_n$ are constants and the operator $\mathcal{L}_n$ is given recursively by the Lenard recursion relation
\begin{equation}\label{eq:Lenard recursion relation}
 \frac{d}{ds}~\mathcal{L}_{n+1}[U] =\left(\frac{d^3}{ds^3}+4U \frac{d}{ds} +2 \frac{dU}{ds} \right)\mathcal{L}_n[U],\quad n\geqslant 1,
\end{equation}
with the initial value $\mathcal{L}_1[U]=U$. It is interesting to note that the system \eqref{eq:b-k} can give us the following Lenard type recursion relation
\begin{equation}\label{eq:recursion relation}
b_{k+1}'=\frac{1}{4}\left(b_k'''-4(2b_1+s)b_k'-2(2b_1+s)'b_k\right), \quad k=1,\ldots, 2m+1,
\end{equation}
with the boundary condition $b_{2m+2}=0$. See also a similar situation where the Painlev\'e III hierarchy is connected to a Lenard type recursion relation in Atkin \cite[Theorem 4.1]{Atkin} and Atkin et al. \cite[Remark 2.1]{ACM}.
\end{rem}

%-------------------------------------------------------------------------------------------------------------------------------------------------------------------------------------------
\subsection*{Asymptotics of the partition function}
With the aid of Theorem \ref{thm:painleve solution}, we next state the asymptotics of the partition function $Z_n(\lambda;\vec{t}\;)$ given in \eqref{def:partiationfunction} in a multiple scaling regime. An essential issue here is a proper and related scalings of the parameters $\lambda$ and $\vec{t}$ in the potential $V(x;\lambda, \vec{t}\;)$. To state the precise assumption, we need a $\phi$-function defined by
\begin{equation}\label{phi}
\phi(z)=2\int_{1}^z\sqrt{x^2-1}dx=2z\sqrt{z^2-1}-\log\left(z+\sqrt{z^2-1}\right), \quad z \in \mathbb{C} \setminus (-\infty,1],
\end{equation}
where the square root and the logarithm  all take the principal branches. Clearly, as $z \to 1$
$$\phi(z) \sim \frac{4\sqrt{2}}{3}(z-1)^{\frac 32}.$$
Hence,
\begin{equation}\label{def:f}
f(z):=\left(\frac{3}{2}\phi(z)\right)^{2/3}
\end{equation}
is analytic near $z=1$, and if $\lambda$ is close to $1$, it is readily seen that
 \begin{equation}\label{def:cjk}
 (f(z)-f(\lambda))^{-j}=\sum_{k=0}^{j}c_{jk}(z-\lambda)^{-k}+O(z-\lambda),\qquad j=1,...,2m,
 \end{equation}
where $c_{jj}=f'(\lambda)^{-j}$ and the other coefficients $c_{jk}$ can also be computed explicitly in terms of the higher order derivatives of $f$ at $z=\lambda$. We now make the following scalings on the parameters $\lambda$ and $\vec{t}$.
\begin{asm}\label{asm}
As $n\to \infty$, it is required that
\begin{itemize}
  \item $\lambda \to 1 $ in such a way that
  \begin{equation}\label{eq:lambdascalingparti}
  2n^{2/3}(\lambda-1)\to s \in \mathbb{R};
\end{equation}

  \item  $\vec{t}\to \vec{0}$ in such a way that
  \begin{equation}\label{eq:tkscaling}
t_k=2\sum_{j=k}^{2m}c_{jk}\tau_j n^{-(1+\frac{2j}{3})}, \quad k=1,...,2m,
\end{equation}
where $c_{jk}$, $j,k=1,\ldots,2m$, is given in \eqref{def:cjk} and  $\vec{\tau}=(\tau_1, \tau_2,\ldots,\tau_{2m})$ is any fixed vector  in $\mathbb{R}^{2m-1}\times (0, +\infty)$.
\end{itemize}
\end{asm}
The condition \eqref{eq:tkscaling} actually means $\vec{t}$ tends to $\vec{0}$ from a specific direction. Moreover, by \eqref{def:cjk} and \eqref{eq:lambdascalingparti}, we have
\begin{equation}\label{eq:tkasy}
t_k\sim 2^{1-k}n^{-1-\frac{2}{3}k}\tau_k, \qquad \textrm{as $n\to \infty$}.
\end{equation}
Our second result is then the following theorem.
\begin{thm} \label{thm:partition asymptotics}
Let $Z_n(\lambda)$ be the partition function \eqref{def:partiationfunction} of the perturbed GUEs \eqref{pGUE}. In the multiple scaling limit when $n\to\infty$ and simultaneously $\lambda\to 1$, $\vec{t} \to \vec{0}$ so that Assumption \ref{asm}
is satisfied, we have
\begin{equation}\label{eq:zn-asy}
Z_n(\lambda)= Z_n^{GUE}e^{2n^{1/3}\tau_1}\exp\left(-\int_{s}^{\infty} \left(b_1(t)(t-s) +\frac{\tau_1}{2\sqrt{t-s}}\right) dt\right)(1+o(1)),
\end{equation}
	where
    \begin{equation} \label{zn-gue}
      Z_n^{GUE}= \frac{(2 \pi)^{n/2}}{ (4n)^{n^2/2}} \prod_{j=1}^n j!
    \end{equation}
    is the partition function for the classical GUE (i.e., $\vec{t}=\vec{0}$ in \eqref{pGUE}),
    and $b_1(s)$ is among the special solutions to the coupled Painlev\'{e} XXXIV system \eqref{eq:b-k} as stated in Theorem \ref{thm:painleve solution}.
\end{thm}
It is readily seen from \eqref{eq:b1asy} that the integral in \eqref{eq:zn-asy} is well-defined, and the asymptotic formula \eqref{eq:zn-asy} depends on the parameters $\tau_k$, $k = 2, \ldots, 2m$ via the function
$b_1(s)=b_1(s;\vec{\tau})$.

\begin{rem}
%The differential identities \eqref{diff identity} and \eqref{diff identity-2}  in the position of pole of the potential \eqref{potential:V} are crucial in our Painlev\'e type asymptotics  of the partition functions  \eqref{def:partiationfunction}. 
We obtain the asymptotics of the partition function  $Z_n(\lambda)$ in terms of  an integral of the function $b_1(s)=b_1(s;\vec{\tau})$, which is the special solution to the 
 coupled Painlev\'{e} XXXIV system \eqref{eq:b-k}.  As given in \eqref{eq:lambdascalingparti}, 
 the lower integration limit $s$ is a proper scaling of $\lambda$, which is the position of the pole of the potential \eqref{potential:V}. 
 The theorem is proved by using certain differential identities with respect to the variable $\lambda$. 
 In the asymptotic study of the perturbed GUE \eqref{model-Bri} with second order pole at the origin and the perturbed LUE \eqref{model-ChenIts} with first order or higher order pole at the origin, people considered differential identities with respect to the coefficients of the pole in the potential instead; see \cite{ACM, Bri:Mezz:Mo,XDZ2015}. In our model, one may also derive differential identities with respect to the coefficients  $\vec{t}$ in the potential \eqref{potential:V}. It would be interesting to see whether there exist any new Painlev\'{e} type system of equations in the the variables $\vec{\tau}$, which are related to  the  coefficients $\vec{t}$. 
This together with the differential identities in $\vec{t}$ may lead to other new integral representation of the asymptotics of the partition function.  We will leave this problem to a further investigation.
  \end{rem}

\begin{rem}
Quite recently, the coupled Painlev\'e systems have appeared frequently in the literature of random matrix theory. For example, in the study of Fredholm determinants associated with the Painlev\'e II or III kernels, the Tracy-Widom type formulas are given in terms of explicit integrals involving a solution to the coupled Painlev\'e II \cite{XD2017} or the Painlev\'e III system \cite{Dai:Xu:Zhang2018}. Moreover, the coupled Painlev\'e II and V systems have also been related to the generating function for the Airy point process and the Bessel point process in \cite{Claeys:Doer2018} and \cite{Char:Doer}, respectively.
\end{rem}

\subsection*{Multiple scaling limit of the correlation kernel}
Finally, we present the multiple scaling limit of the eigenvalue correlation kernel. It comes out that we have found a new multi-parameter family of limiting kernels to describe this local behavior as stated in what follows.

\begin{thm}\label{thm:Correlation kernel-Asy}
Let $K_{n}(x,y;\lambda,\vec{t}\;)$ be the eigenvalue  correlation kernel of the singularly perturbed GUEs given in \eqref{eq:correlation kernel}. Under Assumption \ref{asm}, there exists a multi-parameter family of kernels $K_{\Psi}(u,v;s, \vec{\tau})$ such that
\begin{equation}\label{Cor kernel-asy}
\frac{1}{2n^{2/3}}K_n\left(\lambda+\frac{u}{2n^{2/3}},\lambda+\frac{v}{2n^{2/3}}\right)= K_{\Psi}(u,v;s, \vec{\tau})(1+o(1)),
\end{equation}
uniformly for $u, v$ in any compact subset of $\mathbb{R} \setminus \{0\}$ and for $s$ in compact subset of $\mathbb{R}$.
\end{thm}

The limiting kernels $K_{\Psi}(u,v;s, \vec{\tau})$ are described through the solution of the following special Riemann-Hilbert (RH) problem, which we refer to as the model RH problem for $\Psi$.

\subsubsection*{RH problem for $\Psi$}
\begin{enumerate}
  \item[\rm (a)]  $\Psi(\zeta)=\Psi(\zeta;s, \vec{\tau})$ is a $2 \times 2$ matrix-valued function depending on the parameters $s\in\mathbb{R}$ and $\vec{\tau}=(\tau_1,\tau_2,\ldots,\tau_{2m})\in \mathbb{R}^{2m-1}\times (0,\infty)$, which is analytic for
  $\zeta \in \mathbb{C}\setminus \{\cup^4_{j=1}\Sigma_j\cup\{0\}\}$ with contours $\Sigma_j$, $j=1,2,3,4,$  illustrated in Figure \ref{contour-for-model}.

  \begin{figure}[h]
 \begin{center}
   \includegraphics[width=5.5cm]{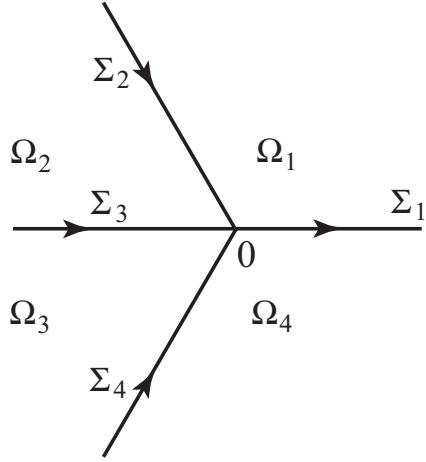} \end{center}
  \caption{The jump contours $\Sigma_j$ and the regions $\Omega_j$, $j=1,2,3,4$, for the RH problem for $\Psi$. Both the sectors $\Omega_2$ and $\Omega_3$ have an opening angle  $\pi/3$.}
 \label{contour-for-model}
\end{figure}

  \item[\rm (b)]  $\Psi$ has limiting values $\Psi_{\pm}(\zeta)$ for $\zeta \in \cup^4_{j=1}\Sigma_j$, where $\Psi_+$ ($\Psi_-$) denotes the limiting values from the left (right) side of $\Sigma_j$, and
 \begin{equation}\label{Psi-jump}
 \Psi_+(\zeta)=\Psi_-(\zeta)
 \left\{
 \begin{array}{ll}
    \begin{pmatrix}
                                 1 & 1 \\
                               0& 1
                                 \end{pmatrix}, &  \zeta \in \Sigma_1, \\[.4cm]
    \begin{pmatrix}
                                0 &1\\
                               -1&0
                                \end{pmatrix},  &  \zeta \in \Sigma_3, \\[.4cm]
    \begin{pmatrix}
                                 1 &0 \\
                                 1 &1
                                 \end{pmatrix}, &   \zeta \in \Sigma_2\cup \Sigma_4.
 \end{array}  \right .  \end{equation}

\item[\rm (c)] As $\zeta\to \infty$, there exists a function $a_1(s;\vec{\tau})$ such that
  \begin{multline}\label{Psi-infty}
 \Psi(\zeta;s, \vec{\tau})=
\begin{pmatrix}
                     1 &  0              \\
                     a_1(s;\vec{\tau}) & 1 \\
\end{pmatrix}
 \left[ I + \frac{
 \Psi_{1}(s;\vec{\tau})}{\zeta} +
O\left(\frac{1}{\zeta^2}\right) \right]
\\
\times e^{-\frac{1}{4}\pi i\sigma_3}\zeta^{-\frac{1}{4} \sigma_3} \frac{I + i \sigma_1}{\sqrt{2}} e^{-\theta(\zeta) \sigma_3},
   \end{multline}
  where
  \begin{equation}\label{def:theta}
  \theta(\zeta):=\theta(\zeta; s) =\frac{2}{3}\zeta^{3/2}+s\zeta^{1/2},
  \end{equation}
  and
 \begin{equation}\label{a1:def}
 (\Psi_1(s;\vec{\tau}))_{12}=a_1(s;\vec{\tau})
 \end{equation}
with $(M)_{ij}$ standing for the $(i,j)$-th entry of a matrix $M$. Here, the branch cuts of the functions $\zeta^{-\frac{1}{4}}, \zeta^{3/2}$ and $\zeta^{1/2}$ are all taken along the negative real axis with $\arg \zeta \in (-\pi, \pi)$, $\sigma_1$ and $\sigma_3$ are the Pauli matrices defined by
\begin{equation}\label{Pauli-matrix}
\sigma_1=\begin{pmatrix}
                     0 &1 \\
                    1 & 0
          \end{pmatrix},
                    \qquad \sigma_3=
                   \begin{pmatrix}
                     1 & 0 \\
                    0 & -1
                    \end{pmatrix}.
 \end{equation}

 \item[\rm (d)] As $\zeta\to 0$, there exists a unimodular matrix  $\Psi_0(s)=\Psi_0(s;\vec{\tau})$, independent of $\zeta$, such that
  \begin{equation}\label{Psi-origin}
  \Psi(\zeta;s,\vec{\tau})=\Psi_0(s)\left[I+O(\zeta)\right] e^{-\sum_{k=1}^{2m}\tau_k\zeta^{-k}\sigma_3}\left\{
  \begin{array}{ll}
    I, &  \zeta \in \Omega_1\cup\Omega_4,
    \\
  \begin{pmatrix}
                                 1 & 0 \\
                               -1& 1
  \end{pmatrix},  &  \zeta \in \Omega_2, \\[.4cm]
     \begin{pmatrix}
                                 1 &0 \\
                                  1 &1
                                  \end{pmatrix}, &   \zeta \in \Omega_3,
 \end{array}  \right .
  \end{equation}
  where the regions $\Omega_1-\Omega_4$ are depicted in Figure \ref{contour-for-model}.
 \end{enumerate}

As we will show later, there exists a unique solution to the above model RH problem. We now set
\begin{equation}\label{def:psi12}
\begin{pmatrix}
    \psi_1(x;s,\vec{\tau}) \\
    \psi_2(x;s,\vec{\tau})
\end{pmatrix}
=\left\{
          \begin{array}{ll}
            \Psi_+(x;s,\vec{\tau})\begin{pmatrix}
                   1 \\
                  1 \\
                 \end{pmatrix}, & \hbox{for $x<0$,} \\[.4cm]
            \Psi_+(x;s,\vec{\tau})\begin{pmatrix}
                   1 \\
                  0 \\
                 \end{pmatrix}, & \hbox{for $x>0$.}
          \end{array}
        \right.
\end{equation}
Then, the limiting kernels $K_{\Psi}$ in Theorem \ref{thm:Correlation kernel-Asy} can be written as
\begin{equation}\label{def:Psi-kernel}
K_{\Psi}(x,y;s,\vec{\tau})=\frac{\psi_1(y;s,\vec{\tau})\psi_2(x;s,\vec{\tau})-\psi_1(x;s,\vec{\tau})\psi_2(y;s,\vec{\tau})}{2\pi i(x-y)}.
\end{equation}

\begin{rem}
The principle of universality (cf. \cite{bi, deift,dkmv1}) suggests that this new family of limiting kernels applies to more general situations whenever the coalescing of the pole and the soft edge of the spectrum occurs, which represents a new universality class.
\end{rem}

\subsection*{About the proofs and organization of the rest of the paper}
The rest of this paper is devoted to the proofs of our results. We deal with the properties of the model RH problem for $\Psi$ in Section \ref{Analysis of the model RH problem}, which include its unique solvability, Lax pair equations and asymptotics as $s\to+\infty$. These results will finally lead to the proof of Theorem \ref{thm:painleve solution} presented at the end of Section \ref{Analysis of the model RH problem}. The proofs of multiple scaling limits of the partition function and the correlation kernel rely on their connections with the classical RH problem that characterize orthogonal polynomials. In Section \ref{sec:connectionRHP}, we recall this RH problem (denoted by $Y$), and establish some new relations between the logarithmic derivative of the partition function (with respect to $\lambda$) and $Y$. We then perform a Deift-Zhou steepest descent analysis of the RH problem for $Y$ in Section \ref{sec:AARHY}. According to Assumption \ref{asm}, the analysis should be performed under conditions \eqref{eq:lambdascalingparti} and \eqref{eq:tkscaling}. It comes out that, in practice, the asymptotics of $Y$ for $\lambda$ in the regime \eqref{eq:lambdascalingparti} alone is not enough for us to derive the asymptotics of the partition function. The reason is what we really obtain from the differential identity is the asymptotics of the logarithmic derivative of the partition function. We then encounter the problem of identifying the integration constant. To resolve this problem, our strategy is the following. In Section \ref{sec:AARHY}, we carry out asymptotic analysis of the RH problem for $Y$ in a larger regime $1+cn^{-2/3}<\lambda<1+dn^{-1/3}$, where $c<0$ and $d>0$ are arbitrarily fixed constants, and the model RH problem $\Psi$ is used in the construction of local parametrix near $z=1$. As a consequence, we are able to prove Theorem \ref{thm:Correlation kernel-Asy} and derive the asymptotics of the logarithmic derivative of the partition function (Lemma \ref{thm:Log-Partition function-asy}), as presented in Section \ref{sec:pfthmkernel}. Particularly, the error bound in the asymptotic formula is uniformly valid for $1+cn^{-2/3}<\lambda<1+dn^{-1/3}$. We then analyze the RH problem for $Y$ with $\lambda>1+n^{-2/5}$ as $n\to \infty$ in Section \ref{sec:AARHY2}, and obtain the asymptotics of the partition function at the end of this section; see Lemma \ref{thm:Partition function-Asy-outer} below. Note that these two ranges of $\lambda$ are overlapped, which enables us to estimate the constant of integration and leads to the proof of Theorem \ref{thm:partition asymptotics} given in Section \ref{sec:pfasyparti}.

%-------------------------------------------------------------------------------------------------------------------------------------
\section{Analysis of the model RH problem}\label{Analysis of the model RH problem}
In this section, we first show that the model RH problem for $\Psi$ is uniquely solvable, and then derive the associated Lax pair equations, whose compatibility condition will give us the coupled Painlev\'{e} XXXIV system. After performing the Deift-Zhou steepest descent analysis to the RH problem for $\Psi$ as $s \to +\infty$, we finally present the proof of Theorem \ref{thm:painleve solution} at the end.

\subsection{Unique solvability of the RH problem for $\Psi$ }

We start with a lemma also known as the vanishing lemma.

\begin{lem}[Vanishing Lemma]
\label{thm:vanishing lemma }
Let $\widetilde{\Psi}^{(1)}(\zeta;s,\vec{\tau})$ with $s\in\mathbb{R}$ and  parameters $\vec{\tau}=(\tau_1, \tau_2,\ldots, \tau_{2m})\in \mathbb{R}^{2m-1}\times (0,\infty)$ be the `homogeneous' version of the RH problem for $\Psi$, i.e., it satisfies items (a), (b) and (d) of the RH problem for $\Psi$, but with the large $\zeta$ behavior replaced by
\begin{equation}\label{Psi-tilde-infty}
 \widetilde{\Psi}^{(1)}(\zeta)=
O\left(\frac{1}{\zeta}\right)\zeta^{-\frac{1}{4} \sigma_3} \frac{I + i \sigma_1}{\sqrt{2}} e^{-\theta(\zeta) \sigma_3},
   \end{equation}
where $\theta$ is defined in \eqref{def:theta}. Then, the solution is trivial, that is,
$$\widetilde{\Psi}^{(1)}(\zeta)\equiv 0.$$
\end{lem}
\begin{proof}
We first bring all the jumps of $\widetilde{\Psi}^{(1)}$ to the real axis and remove the exponential term in its large $\zeta$ behavior by introducing the following transformation
\begin{equation}\label{def:psi-tilde-2}
 \widetilde{\Psi}^{(2)}(\zeta): = \widetilde{\Psi}^{(1)}(\zeta)e^{\theta(\zeta) \sigma_3}
 \left\{
 \begin{array}{ll}
    I, &  \textrm{for $\zeta \in \Omega_1\cup \Omega_4$,}
    \\
    \begin{pmatrix}
                                 1 &0 \\
                                 e^{2\theta(\zeta)} &1
                                 \end{pmatrix}, &   \textrm{for $\zeta \in \Omega_2$,} \\[.4cm]
    \begin{pmatrix}
                                 1 &0 \\
                                 -e^{2\theta(\zeta)} &1
                                 \end{pmatrix}, &   \textrm{for $\zeta \in \Omega_3$,}
 \end{array}  \right .
 \end{equation}
where the regions $\Omega_1-\Omega_4$ are depicted in Figure \ref{contour-for-model}.
Then, it is easily seen that $\widetilde{\Psi}^{(2)}(\zeta)$ satisfies the following RH problem.
\subsubsection*{RH problem for $\widetilde{\Psi}^{(2)}$}
\begin{enumerate}
\item[\rm (a)] $\widetilde{\Psi}^{(2)}(\zeta)$ is defined and analytic in
  $\mathbb{C}\setminus \mathbb{R}$.
\item[\rm (b)] $\widetilde{\Psi}^{(2)}(\zeta)$ satisfies the jump condition
\begin{equation}\label{psi-tilde-2-jump}
 \widetilde{\Psi}^{(2)}_+(x)= \widetilde{\Psi}^{(2)}_-(x)
 \left\{
 \begin{array}{ll}
    \begin{pmatrix}
                                 1 &e^{-2\theta(x)} \\
                                 0 &1
                                 \end{pmatrix}, &   x>0, \\[.4cm]
    \begin{pmatrix}
                                 e^{2\theta_+(x)} &1 \\
                                 0 &e^{2\theta_-(x)}
                                 \end{pmatrix}, &   x<0.
 \end{array}  \right .  \end{equation}
\item[\rm (c)] As $\zeta\to\infty$, we have
\begin{equation}\label{psi-tilde-2-infty}
 \widetilde{\Psi}^{(2)}(\zeta)=
O\left(\frac{1}{\zeta^{3/4}}\right).
   \end{equation}
\item[\rm (d)]As $\zeta\to 0 $, we have
\begin{equation}\label{psi-tilde-2-origin}
 \widetilde{\Psi}^{(2)}(\zeta)=O(1)
  e^{-\sum_{k=1}^{2m}\tau_k\zeta^{-k}\sigma_3}.
  \end{equation}
\end{enumerate}

We next define a matrix-valued function $M$ by
\begin{equation}\label{def:M}
M(\zeta):= \widetilde{\Psi}^{(2)}(\zeta)
    \begin{pmatrix}
    0 &-1
    \\
    1 &0
    \end{pmatrix}\left(\widetilde{\Psi}^{(2)}(\bar{\zeta})\right)^*,~~\zeta\in \mathbb{C}\setminus \mathbb{R},
 \end{equation}
where $\bar{\zeta}$  denotes the conjugate of $\zeta$ and $A^*$ stands for the Hermitian conjugate of the matrix $A$.
From the RH problem for $ \widetilde{\Psi}^{(2)}$, it is readily seen that $M(\zeta)$ is analytic $\mathbb{C} \setminus \mathbb{R}$. Moreover,
$M(\zeta)$ is bounded near the origin and $M(\zeta)=O\left(\frac{1}{\zeta^{3/2}}\right)$ as $\zeta\to \infty$. Thus, by Cauchy's theorem, we have
\begin{equation}\label{integral-1}
\int_{\mathbb{R}} M_+(x)dx=0.
 \end{equation}
Using the jump condition \eqref{psi-tilde-2-jump} and the definition of $M$ in \eqref{def:M}, the integral \eqref{integral-1} can be rewritten as
\begin{multline}\label{integral-2}
\int_{-\infty}^0 \widetilde{\Psi}^{(2)}_-(x)
    \begin{pmatrix}
                                1 &-e^{2\theta_+(x)} \\
                                 e^{2\theta_-(x)} &0
                                 \end{pmatrix}\left(\widetilde{\Psi}^{(2)}_-(x)\right)^*dx
                                 \\
                                +\int^{\infty}_0 \widetilde{\Psi}^{(2)}_-(x)dx
    \begin{pmatrix}
                                e^{2\theta(x)}&-1 \\
                                 1 &0
                                 \end{pmatrix}\left(\widetilde{\Psi}^{(2)}_-(x)\right)^*dx=0.
 \end{multline}
By adding this relation to its Hermitian conjugate, we have
 \begin{multline}\label{integral-3}
\int_{-\infty}^0 \widetilde{\Psi}^{(2)}_-(x)
    \begin{pmatrix}
                                2 &0\\
                                0 &0
                                 \end{pmatrix}\left(\widetilde{\Psi}^{(2)}_-(x)\right)^*dx  +\int^{\infty}_0 \widetilde{\Psi}^{(2)}_-(x)
    \begin{pmatrix}
                                2 e^{2\theta(x)}&0 \\
                                 0 &0
                                 \end{pmatrix}\left(\widetilde{\Psi}^{(2)}_-(x)\right)^*dx=0,
 \end{multline}
where use has been made of the fact that $\theta_+(x)$ is purely imaginary and
$$\theta_+(x)=-\theta_-(x),\qquad x<0.$$
Thus, we obtain from \eqref{integral-3} that the first column of $\widetilde{\Psi}^{(2)}_-(x)$ vanishes for real value of $x$, which also implies the first column of $\widetilde{\Psi}^{(2)}$ vanishes in the lower half complex plane. By the jump relation \eqref{psi-tilde-2-jump}, the second column of $\widetilde{\Psi}^{(2)}$ vanishes in the upper half complex plane. It then follows from the Carlson's theorem that the other entries of $\widetilde{\Psi}^{(2)}$  vanish in the complex plane as well, cf. \cite{ik,xz2011}. Hence, on account of \eqref{def:psi-tilde-2}, we arrive at $\widetilde{\Psi}^{(1)}(\zeta)\equiv 0$.

This completes the proof of Lemma \ref{thm:vanishing lemma }.
\end{proof}

By a standard analysis  \cite{dkmv1,dkmv2,fikn}, the following proposition is immediate.
\begin{prop}\label{thm:exsitence}
There exists a unique solution to the RH problem for $\Psi$ for the parameters $s\in\mathbb{R}$ and $\vec{\tau}=(\tau_1, \tau_2,\ldots, \tau_{2m})\in\mathbb{R}^{2m-1}\times (0,\infty)$.
\end{prop}
%----------------------------------------------------------------------------------------------------------------------------------------------------

\subsection{Lax pair equations and the coupled Painlev\'{e}  system}

We next derive the Lax pair for $\Psi$ and establish its connection to the coupled Painlev\'{e} system \eqref{eq:b-k} from the associated compatibility condition.
\begin{prop}\label{thm:laxpair}
Let $\Psi=\Psi(\zeta;s, \vec{\tau})$ be a solution of the model RH problem. Then, we have the following Lax pair:
    \begin{align}
    \frac{\partial\Psi}{\partial \zeta} & =  \left(\sum_{k=0}^{2m+1}\frac{A_k}{\zeta^k} +\zeta\sigma_{-} \right) \Psi, \label{Lax pair-1} \\[.1cm]
    \frac{\partial\Psi}{\partial s} &= B(\zeta;s)\Psi, \label{Lax pair-2}
  \end{align}
  where $\sigma_{-}=
               \begin{pmatrix}
                 0 & 0 \\
                 1 & 0\\
               \end{pmatrix}$,
  the coefficient matrices take the form
  \begin{align} \label{A-0}
      A_0(s) & = \begin{pmatrix}
                0 & 1 \\
                b_1+s& 0
              \end{pmatrix},
              \\
       A_k(s) &= \begin{pmatrix}
    \frac{b_k'}{2} & -b_k \\
    \frac{b_k''}{2}-(2b_1+s)b_k-b_{k+1} & - \frac{b_k'}{2}
  \end{pmatrix},\qquad k=1,2,\ldots,2m+1,
  \label{A-k}
  \end{align}
 with $' =\frac{d}{ds}$, and
  \begin{equation} \label{psi-B}
    B(\zeta;s) = \begin{pmatrix}
    0 &1\\ \zeta+2b_1+s &0
  \end{pmatrix}.
  \end{equation}
Moreover, the functions $b_1(s),\ldots,b_{2m+1}(s)$ in \eqref{A-0} and \eqref{A-k} satisfy
the coupled Painlev\'{e} XXXIV system \eqref{eq:b-k} with the parameters $\tilde \tau_p$ given by \eqref{def:tauptilde}.
\end{prop}
\begin{proof}
Since the jumps in the RH problem for $\Psi$ are constant matrices,  it follows that the functions
\begin{equation}
A(\zeta;s):=\frac{\partial \Psi}{\partial \zeta}\cdot \Psi^{-1}, \qquad B(\zeta;s):=\frac{\partial \Psi}{\partial s}\cdot \Psi^{-1}
\end{equation}
are meromorphic in $\zeta$ with possible isolated singularity at the origin. From the asymptotic behaviors of $\Psi$ near $\zeta=\infty$ and $\zeta=0$ as given in \eqref{Psi-infty}--\eqref{Psi-origin}, we have
that
\begin{align}\label{eq:Aexp}
A(\zeta;s)&=\sum_{k=0}^{2m+1}\frac{A_k}{\zeta^k}+\zeta\sigma_-,
\\
B(\zeta;s)&=
\begin{pmatrix}
0 & 1 \\
\zeta+b_1(s)+a_1'(s)+\frac {s}{2} & 0
       \end{pmatrix},
\end{align}
where
\begin{align}\label{eq:A-0}
 A_0&=\begin{pmatrix}
         0 & 1 \\
         b_1(s)+s & 0
 \end{pmatrix},
\\
\label{A-1}
b_1(s) & = -(A_1)_{12}=a_1(s)^2-2a_2(s)-\frac{s}{2}
 \end{align}
with
\begin{equation}
a_2(s)=a_2(s;\vec{\tau})=(\Psi_1(s;\vec{\tau}))_{11}.
\end{equation}

To show that the other $\zeta$-independent matrices $A_k$, $k=1,\ldots,2m+1$ in \eqref{eq:Aexp} have the explicit expressions as given in \eqref{A-k}, we note that the compatibility condition
$$\frac{\partial^2 \Psi}{\partial \zeta \partial s}=\frac{\partial^2 \Psi}{\partial s \partial \zeta }$$
for the differential equations \eqref{Lax pair-1} and \eqref{Lax pair-2} is the zero curvature
relation
\begin{equation}
  \frac{\partial A}{\partial s} -\frac{ \partial B}{\partial \zeta} +[A,B] = 0,
\end{equation}
where $[L,K] = LK - KL$ stands for the standard commutator of two matrices. Hence, it follows that
\begin{align}\label{eq:zerocurv}
&\sum_{k=1}^{2m+1}\frac{A_k'}{\zeta^k}
+
\begin{pmatrix}
0 & 0
\\
b_1'(s) & 0
\end{pmatrix}
\nonumber
\\
&=
\begin{pmatrix}
0 & 1 \\
\zeta+b_1(s)+a_1'(s)+\frac {s}{2} & 0
\end{pmatrix}
\left(\sum_{k=1}^{2m+1}\frac{A_k}{\zeta^k}+\begin{pmatrix}
0 & 1 \\
\zeta+b_1(s)+s & 0
\end{pmatrix}\right)
\nonumber
\\
&\qquad -\left(\sum_{k=1}^{2m+1}\frac{A_k}{\zeta^k}+\begin{pmatrix}
0 & 1 \\
\zeta+b_1(s)+s & 0
\end{pmatrix}\right)\begin{pmatrix}
0 & 1 \\
\zeta+b_1(s)+a_1'(s)+\frac {s}{2} & 0
\end{pmatrix}.
\end{align}
In addition, since $\det \Psi(\zeta)=1$, we have $\mathrm{Tr}A=0$, which means
\begin{equation}\label{Trace zero}
\mathrm{Tr}A_k=0.
 \end{equation}
As a consequence, if one sets
$$b_k=-(A_k)_{12}, \qquad k=1,\ldots, 2m+1,$$
by comparing the coefficients of $\zeta^j$, $j=0,\ldots,-2m-1$, on both sides of \eqref{eq:zerocurv} and making use of \eqref{Trace zero}, we obtain the relation
\begin{equation} \label{a-1-b-1}
   a_1'=b_1+\frac {s}{2},
\end{equation}
and
 \begin{equation}\label{eq:bi-a-i}
 \begin{cases}
   (A_k)_{11}=-(A_k)_{22}=\frac{b_k'(s)}{2}, \\
          (A_k)_{21}=\frac{b_k''(s)}{2}-(2b_1(s)+s)b_k(s)-b_{k+1}(s),\\
          (A_k)_{21}'=2(2b_1+s)(A_k)_{11}+2(A_{k+1})_{11},
 \end{cases}
 \end{equation}
where $k=1,\ldots, 2m$ and $b_{k}=0$ for $k>2m+1$, as shown in \eqref{A-k}. Substituting the first two equations into the third one, we obtain
 the Lenard type recursion relation \eqref{eq:recursion relation} for $b_k$, $k=1,\ldots, 2m$.

To derive the coupled Painlev\'{e}  system \eqref{eq:b-k}, we observe from the asymptotic behavior of $\Psi$ near $\zeta=0$ (see \eqref{Psi-origin}) that as $\zeta \to 0$,
\begin{equation}\label{A-det}
\det A(\zeta;s)=-\sum_{p=2m+2}^{4m+2}\tilde{\tau}_{p}\zeta^{-p}+O(\zeta^{-(2m+1)}),
 \end{equation}
where the constants $\tilde{\tau}_{p}$ are given in \eqref{def:tauptilde}.
The above formula, together with \eqref{eq:Aexp}, \eqref{A-0} and \eqref{A-k}, implies \eqref{eq:b-k}.

This completes the proof of Proposition \ref{thm:laxpair}.
\end{proof}

\begin{rem}
  We derive the coupled Painlev\'{e} XXXIV  system \eqref{eq:b-k} from \eqref{A-det}. One may also obtain it from the Lenard type recursion relation \eqref{eq:recursion relation} by using a similar argument as in Atkin \cite[Theorem 4.1]{Atkin}.
\end{rem}

%---------------------------------------------------------------------------------------------------------------------------------------------------
\subsection{Asymptotic analysis of the RH problem for $\Psi$ as $s \to +\infty$}\label{sec-Psi-large-S}

In this section, we shall perform the Deift-Zhou steepest descent analysis to the RH problem for $\Psi$ as $s \to +\infty$, which will be essential in proving the asymptotics of $b_1$ shown in \eqref{eq:b1asy}. It consists of a series of explicit and invertible transformations which leads to an RH problem tending to the identity matrix as $s \to +\infty$.

\subsubsection{$\Psi\to U$: Rescaling}
Define
\begin{equation}\label{psi to U}
U(\zeta;s,\vec{\tau})=\Psi(s\zeta;s,\vec{\tau}).
\end{equation}
It is then straightforward to show that the function $U$ satisfies the following RH problem.

\subsubsection*{RH problem for $U$}
\begin{enumerate}
\item[\rm (a)] $U(\zeta)$ is defined and analytic in
  $\mathbb{C}\setminus \{\cup^4_{j=1}\Sigma_j\cup\{0\}\}$, where the contours $\Sigma_j$, $j=1,2,3,4$  are illustrated in Figure \ref{contour-for-model}.
\item[\rm (b)] $U$ shares the same piecewise-constant jump condition as $\Psi(\zeta)$; see \eqref{Psi-jump}.
\item[\rm (c)] As $\zeta\to\infty$, we have
\begin{multline}\label{U-infty}
 U(\zeta;s, \vec{\tau})= \left(
                   \begin{array}{cc}
                     1 & 0\\
                     a_1(s;\vec{\tau}) & 1 \\
                   \end{array}
                 \right)
 \left[ I + \frac{
 \Psi_{1}(s; \vec{\tau})}{s\zeta} +
O\left(\frac{1}{\zeta^2}\right) \right]
 \\
 \times e^{-\frac{1}{4}\pi i\sigma_3}(s\zeta)^{-\frac{1}{4} \sigma_3} \frac{I + i \sigma_1}{\sqrt{2}} e^{-s^{3/2}\left(\frac23 \zeta^{3/2}+ \zeta^{1/2} \right) \sigma_3},
   \end{multline}
  where the functions $a_1$ and $\Psi_1$ are given in \eqref{Psi-infty}.
\item[\rm (d)]As $\zeta\to 0 $, we have
\begin{equation}\label{U-origin}
 U(\zeta)=\Psi_0(s)\left(I+O(\zeta)\right) e^{-\sum_{k=1}^{2m}\tau_ks^{-k}\zeta^{-k}\sigma_3}\left\{
  \begin{array}{ll}
    I, &  \zeta \in \Omega_1\cup\Omega_4,
    \\
 \begin{pmatrix}
                                 1 & 0 \\
                               -1& 1
                                \end{pmatrix},  &  \zeta \in \Omega_2, \\[.4cm]
    \begin{pmatrix}
                                 1 &0 \\
                                  1 &1
                                 \end{pmatrix}, &   \zeta \in \Omega_3,
 \end{array}  \right .
  \end{equation}
where the function $\Psi_0$ is given in \eqref{Psi-origin} and the regions $\Omega_1-\Omega_4$ are depicted in Figure \ref{contour-for-model}.
\end{enumerate}

\subsubsection{$U \to W$: Contour deformation}

In the second transformation we apply contour deformation. The rays $\Sigma_2$ and  $\Sigma_4$
emanating from the origin are replaced by their parallel lines $\widetilde{\Sigma}_2$ and  $\widetilde{\Sigma}_4$ emanating from the point $-1$. These lines divide the whole complex plane into six regions, which we denote by $\Omega_1,\ldots,\Omega_6$; see Figure \ref{contour-for-V} for an illustration.

We then define
\begin{equation}\label{U to W}
W(\zeta)=\left\{\begin{array}{ll}
                 U(\zeta) & \mbox{for} \quad  \zeta \in \cup_{j=1}^4\Omega_j, \\
                 U(\zeta) \begin{pmatrix}
                             1 & 0 \\
                             1 & 1
                            \end{pmatrix}& \mbox{for} \quad \zeta \in \Omega_5,\\[.4cm]
                  U(\zeta)  \begin{pmatrix}
                             1 & 0 \\
                             -1 & 1
                           \end{pmatrix} & \mbox{for} \quad  \zeta \in \Omega_6.
                \end{array}
\right.
\end{equation}
\begin{figure}[h]
 \begin{center}
   \includegraphics[width=7.5cm]{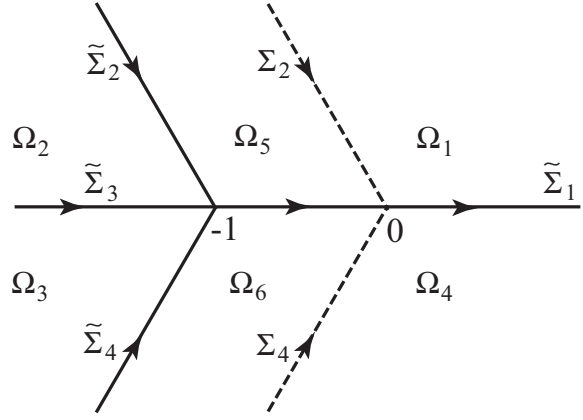} \end{center}
  \caption{The contours $\widetilde{\Sigma}_j$, $j=1,\ldots,4$, and the domains $\Omega_k$, $k=1,\ldots,6$.}
 \label{contour-for-V}
\end{figure}
It is readily seen that the function $W$ defined above satisfies the following conditions.

\subsubsection*{RH problem for $W$}
\begin{enumerate}
\item[\rm (a)] $W(\zeta)$ is analytic in $\mathbb{C} \setminus
\{\cup^4_{j=1}\widetilde{\Sigma}_j\cup\{-1\}\}$, where the contours $\widetilde{\Sigma}_j$, $j=1,2,3,4$, are shown as the solid lines in Figure \ref{contour-for-V}. Note that
\begin{equation}
\widetilde{\Sigma}_1=(-1,+\infty), \qquad \widetilde{\Sigma}_3=(-\infty,-1).
\end{equation}
\item[\rm (b)] $W(\zeta)$ satisfies the jump condition
\begin{equation}\label{W-jump}
 W_+(\zeta)=W_-(\zeta)
 \left\{
 \begin{array}{ll}
    \begin{pmatrix}
                                 1 & 1 \\
                               0& 1
                                 \end{pmatrix}, &  \zeta \in \widetilde{\Sigma}_1, \\[.4cm]
    \begin{pmatrix}
                                0 &1\\
                               -1&0
                                \end{pmatrix},  &  \zeta \in \widetilde{\Sigma}_3, \\[.4cm]
    \begin{pmatrix}
                                 1 &0 \\
                                 1 &1
                                 \end{pmatrix}, &   \zeta \in \widetilde{\Sigma}_2 \cup \widetilde{\Sigma}_4.
 \end{array}  \right .  \end{equation}
\item[\rm (c)] As $\zeta\to \infty$, $W$ has the same asymptotic behavior as $U$.
\item[\rm (d)]As $\zeta\to 0 $, we have
\begin{equation}\label{W-origin}
W(\zeta)=\Psi_0(s)\left(I+O(\zeta)\right) e^{-\sum_{k=1}^{2m}\tau_ks^{-k}\zeta^{-k}\sigma_3}.
\end{equation}
\end{enumerate}

All the conditions in the above RH problem are straightforward to check except the jump condition on $(-1,0)$, which we verify now. By \eqref{U to W} and item (b) in the RH problem for $U$, we have, if $\zeta\in(-1,0)$,
\begin{align*}
W_+(\zeta)&=U_{+}(\zeta)\begin{pmatrix}
                             1 & 0 \\
                             1 & 1 \\
                            \end{pmatrix}
         =U_{-}(\zeta)\begin{pmatrix}
                             0 & 1 \\
                             -1 & 0 \\
                            \end{pmatrix}\begin{pmatrix}
                             1 & 0 \\
                             1 & 1 \\
                            \end{pmatrix}
\\
&=W_{-}(\zeta)\begin{pmatrix}
                             1 & 0 \\
                             1 & 1 \\
                            \end{pmatrix}\begin{pmatrix}
                             0 & 1 \\
                             -1 & 0 \\
                            \end{pmatrix}\begin{pmatrix}
                             1 & 0 \\
                             1 & 1 \\
                            \end{pmatrix}
=W_{-}(\zeta)\begin{pmatrix}
                             1 & 1 \\
                             0 & 1 \\
                            \end{pmatrix},
\end{align*}
as shown in \eqref{W-jump}.

\subsubsection{$ W \to Q$: Normalization at $\infty$ and $0$}
Define
\begin{equation}\label{def:g}
g(\zeta)=\frac{2}{3}(\zeta+1)^{\frac{3}{2}}, \qquad \zeta\in\mathbb{C}\setminus(-\infty, -1],
\end{equation}
where $\arg(\zeta+1)\in(-\pi,\pi)$. It is easy to see that
\begin{equation}\label{eq:asyofg}
g(\zeta)= \frac{2}{3}\zeta^{3/2}+\zeta^{1/2}+ \frac {1}{4\zeta^{1/2}}+O(1/\zeta^{3/2}), \qquad \zeta \to \infty,
\end{equation}
and
\begin{equation}\label{eq:gjump}
g_{+}(\zeta)+g_-(\zeta)=0, \qquad \zeta<-1.
\end{equation}
We also define
\begin{equation}\label{def:q}
q(\zeta)=\sqrt{\zeta+1}\sum_{k=1}^{2m}\frac{c_k}{s^{k}\zeta^k}, \qquad \zeta\in\mathbb{C}\setminus(-\infty, -1],
\end{equation}
with $\arg(\zeta+1)\in(-\pi,\pi)$, where the coefficients $c_k=c_k(s)$ are chosen such that as $\zeta\to 0$
\begin{equation}\label{q-zeta-approx}
q(\zeta)=\sum_{k=1}^{2m}\frac{\tau_k}{s^{k}\zeta^k}+O(1).
\end{equation}
Again, we have that
\begin{equation}\label{eq:qjump}
q_{+}(\zeta)+q_-(\zeta)=0, \qquad \zeta<-1.
\end{equation}
Note that
$$ \sqrt{\zeta  +1}=\sum_{k=0}^\infty \frac{(-1)^k(-\frac{1}{2})_k}{k!}\zeta^k, \qquad \textrm{for } |\zeta| < 1$$
with $(a)_k=\frac{\Gamma(a+k)}{\Gamma(a)}=a(a+1)\ldots(a+k-1)$ being the Pochhammer symbol. It is readily seen that the constants $c_k$, $k=1,\ldots,2m$, in \eqref{def:q} satisfy the following linear system
\begin{equation}\label{eq:ck}
\sum_{j=i}^{2m} (-1)^{j-i}\frac{(-\frac12)_{j-i}}{(j-i)!}\frac{c_j}{s^j}=\frac{\tau_i}{s^i},\qquad i=1,\ldots,2m.
\end{equation}
Since the coefficient matrix in \eqref{eq:ck} is an upper triangular matrix, one can determine these constants recursively to obtain that
\begin{equation}
c_{2m}=\tau_{2m},\quad c_{2m-1}=\left(\frac{\tau_{2m-1}}{s^{2m-1}}-\frac{\tau_{2m}}{2s^{2m}}\right)s^{2m-1}, \quad \ldots.
\end{equation}
In particular, we have
\begin{equation}\label{eq:c1asy}
c_1=\tau_1+O(1/s), \qquad s\to+\infty.
\end{equation}

The third transformation is then defined by
\begin{equation}\label{W to Q}
Q(\zeta)=\begin{pmatrix}
             1 & 0\\
             -i\left(\frac{s^{3/2}}{4}+\frac {c_1}{ s }
             \right) & 1 \\
           \end{pmatrix}
s^{\frac{1}{4}\sigma_3}e^{\frac{1}{4}\pi i\sigma_3}
\begin{pmatrix}
1 & 0 \\
-a_1(s) & 1 \\
\end{pmatrix}
W(\zeta)e^{(s^{3/2}g(\zeta)+q(\zeta))\sigma_3}.
\end{equation}
By \eqref{eq:gjump} and \eqref{eq:qjump}, it is readily to verify that $Q$ satisfies the following RH problem.

\subsubsection*{RH problem for $Q$}
\begin{enumerate}
\item[\rm (a)] $Q(\zeta)$ is defined and analytic in $\mathbb{C} \setminus
\{\cup^4_{j=1}\widetilde{\Sigma}_j\cup\{-1\}\}$.

\item[\rm (b)] $Q(\zeta)$  satisfies the jump condition
\begin{equation}\label{Q-jump}
 Q_+(\zeta)=Q_-(\zeta)J_Q(\zeta)=Q_-(\zeta)
 \left\{
 \begin{array}{ll}
\begin{pmatrix}
1 & e^{-2(s^{3/2}g(\zeta)+q(\zeta))} \\
0 & 1
\end{pmatrix}, &  \zeta \in \widetilde{\Sigma}_1, \\[.4cm]
    \begin{pmatrix}
                                1 & 0\\
                               e^{2(s^{3/2}g(\zeta)+q(\zeta))} & 1
    \end{pmatrix},  &  \zeta \in \widetilde{\Sigma}_2\cup \widetilde{\Sigma}_4, \\[.4cm]
    \begin{pmatrix}
                                0 & 1\\
                               -1 & 0
    \end{pmatrix},  & \zeta \in \widetilde{\Sigma}_3.
\end{array}  \right .
\end{equation}

\item[\rm (c)] As $\zeta\to\infty$, we have
\begin{equation}\label{Q: infty}
Q(\zeta)=
 \left( I +\frac{Q_1}{\zeta}+ O\left(\frac{1}{\zeta^2}\right)
 \right)\zeta^{-\frac{1}{4} \sigma_3} \frac{I + i \sigma_1}{\sqrt{2}}
   \end{equation}
with
\begin{equation}\label{def:Q-12}
(Q_1)_{12}=i\left(\frac{a_1(s)}{s^{1/2}}-\frac{s^{3/2}}{4}-\frac {c_1}{s}\right).
\end{equation}
\item[\rm (d)]$Q(\zeta)$ is bounded near the origin.
\end{enumerate}

For the convenience of the reader, we give a proof of \eqref{Q: infty} and \eqref{def:Q-12} in what follows. By \eqref{U-infty} and \eqref{U to W}, it follows that,
as $\zeta \to \infty$,
\begin{multline}\label{W-infty2}
 W(\zeta)= \begin{pmatrix}
                     1 & 0\\
                     a_1(s) & 1 \\
 \end{pmatrix}e^{-\frac{1}{4}\pi i\sigma_3}s^{-\frac{1}{4}\sigma_3}
 \left[ I + \frac{
 \widehat{\Psi}_{1}(s)}{\zeta} +
O\left(\frac{1}{\zeta^2}\right) \right]
 \\
 \times \zeta^{-\frac{1}{4} \sigma_3} \frac{I + i \sigma_1}{\sqrt{2}} e^{-s^{3/2}\left(\frac23 \zeta^{3/2}+ \zeta^{1/2} \right) \sigma_3},
   \end{multline}
where
\begin{equation}
\widehat{\Psi}_{1}(s)=\frac{1}{s}
\begin{pmatrix}
\ast & i s^{1/2}a_1(s)
\\
\ast & \ast
\end{pmatrix}
\end{equation}
with $\ast$ being certain unimportant entries. In view of \eqref{eq:asyofg} and \eqref{def:q}, we note that
\begin{multline}
e^{-s^{3/2}\left(\frac23 \zeta^{3/2}+ \zeta^{1/2}\right)\sigma_3}e^{(s^{3/2}g(\zeta)+q(\zeta))\sigma_3}
\\
=\left[I+\begin{pmatrix}
\frac{c_1}{s}+\frac{s^{3/2}}{4} & 0
\\
0 & -\frac{c_1}{s}-\frac{s^{3/2}}{4}
\end{pmatrix}\frac{1}{\zeta^{1/2}}+O\left(\frac{1}{\zeta^{3/2}}\right)
\right]
\end{multline}
for large $\zeta$. Inserting the above formula into \eqref{W-infty2}, it follows from an elementary calculation that
\begin{multline}
W(\zeta)e^{(s^{3/2}g(\zeta)+q(\zeta))\sigma_3}=\begin{pmatrix}
                     1 & 0\\
                     a_1(s) & 1 \\
 \end{pmatrix}e^{-\frac{1}{4}\pi i\sigma_3}s^{-\frac{1}{4}\sigma_3}\begin{pmatrix}
             1 & 0\\
             i\left(\frac{s^{3/2}}{4}+\frac {c_1}{ s }
             \right) & 1 \\
           \end{pmatrix}
 \\
 \times
 \left[ I+\begin{pmatrix}
             \ast & i\left(\frac{a_1(s)}{s^{1/2}}-\frac{s^{3/2}}{4}-\frac {c_1}{s}\right)\\
             \ast & \ast
           \end{pmatrix}
           \frac{1}{\zeta} +
O\left(\frac{1}{\zeta^2}\right) \right]
\zeta^{-\frac{1}{4} \sigma_3} \frac{I + i \sigma_1}{\sqrt{2}}.
\end{multline}
This, together with \eqref{W to Q}, gives us \eqref{Q: infty} and \eqref{def:Q-12}.

\subsubsection{Outer parametrix}
By \eqref{def:g}, we see that for sufficiently large positive $s$,
 the jump matrices of $Q$ tend to the identity matrix exponentially fast except the one on $\widetilde{\Sigma}_3=(-\infty,-1)$.
Thus, we expect that $Q$ should be approximated by the solution to the following RH problem.

\subsubsection*{RH problem for $Q^{(\infty)}$}
\begin{enumerate}
\item[\rm (a)] $Q^{(\infty)}(\zeta)$ is defined and analytic in $\mathbb{C} \setminus
(-\infty,-1].$

\item[\rm (b)] $Q^{(\infty)}(\zeta)$  satisfies the jump condition
\begin{equation}\label{Q-out-jump}
 Q^{(\infty)}_+(\zeta)=Q^{(\infty)}_-(\zeta)
       \begin{pmatrix}
       0 & 1\\
       -1 & 0
       \end{pmatrix}
                             ,\qquad \zeta\in(-\infty, -1).
    \end{equation}

\item[\rm (c)] As $\zeta\to\infty$, we have
\begin{equation}
Q^{(\infty)}(\zeta)=
 \left( I + O\left(\frac{1}{\zeta}\right)
 \right)\zeta^{-\frac{1}{4} \sigma_3} \frac{I + i \sigma_1}{\sqrt{2}}.
   \end{equation}
\end{enumerate}

The solution to the above RH problem is explicitly given by
\begin{equation}\label{Out solution}
Q^{(\infty)}(\zeta)=
(\zeta+1)^{-\frac 14\sigma_3} \frac{I + i \sigma_1}{\sqrt{2}},
\end{equation}
where the branch is chosen as $\arg(\zeta+1)\in (-\pi, \pi)$.

\subsubsection{Local parametrix near $\zeta=-1$}\label{subsubsec:Airypara}

Near $\zeta=-1$, the outer parametrix $Q^{(\infty)}(\zeta)$ is no longer a good approximation to $Q(\zeta)$.  We seek
a parametrix $Q^{(-1)}(\zeta)$ satisfying the following RH problem:

\subsubsection*{RH problem for $Q^{(-1)}$}
\begin{enumerate}
\item[\rm (a)] $Q^{(-1)}(\zeta)$ is  analytic in  $\overline{U(-1,r)} \setminus
\{\cup^4_{j=1}\widetilde{\Sigma}_j\cup\{-1\}\}$, where $U(a,b):=\{\zeta | \zeta\in\mathbb{C}, |\zeta-a|<b\}$.
\item[\rm (b)] $Q^{(-1)}(\zeta)$ satisfies the same jump condition \eqref{Q-jump} as $Q$ for $\zeta \in U(-1,r) \cap \{\cup^4_{j=1}\widetilde{\Sigma}_j\}$.

\item[\rm (c)] As $s \to +\infty$, $Q^{(-1)}(\zeta)$ matches $Q^{(\infty)}(\zeta)$ on the boundary of $U(-1,r)$, i.e.,
\begin{equation}\label{Matching condition}
Q^{(-1)}(\zeta)=
(I+O(1/s))Q^{(\infty)}(\zeta), \qquad \zeta\in \partial U(-1,r).
\end{equation}
\end{enumerate}
The construction of $Q^{(-1)}$ is standard (cf. \cite{deift,dkmv2}) with the aid of the so-called Airy parametrix
$\Phi_{\Ai}$ defined by
\begin{equation}\label{Airy-model-solution}
  \Phi_{\Ai}(\zeta)=M_{\Ai}\left\{
                 \begin{array}{ll}
                 \begin{pmatrix}
                        \Ai(\zeta) & \Ai(\omega^2 \zeta) \\
                   \Ai'(\zeta)&\omega^2 \Ai'(\omega^2 \zeta)
                   \end{pmatrix} e^{-\frac{\pi i}{6}\sigma_3}, & \textrm{for $\zeta\in \Omega_1$,} \\[.4cm]
                      \begin{pmatrix}
                        \Ai(\zeta) & \Ai(\omega^2 \zeta) \\
                      \Ai'(\zeta)&\omega^2 \Ai'(\omega^2\zeta)\end{pmatrix}e^{-\frac{\pi i}{6}\sigma_3}
                      \begin{pmatrix}
                                                    1 & 0 \\
                                                        -1 & 1\end{pmatrix}
                 , & \textrm{for $\zeta\in \Omega_2$,} \\[.4cm]
                           \begin{pmatrix}
                         \Ai(\zeta) & -\omega^2 \Ai(\omega\zeta) \\
                        \Ai'(\zeta)&- \Ai'(\omega \zeta) \end{pmatrix}e^{-\frac{\pi i}{6}\sigma_3}
                        \begin{pmatrix}
                                                       1 & 0 \\
                                                      1 & 1 \end{pmatrix}
                 , &\textrm{for $\zeta\in \Omega_3$,} \\[.4cm]
                       \begin{pmatrix}
                    \Ai(\zeta) & -\omega^2 \Ai(\omega \zeta) \\
                        \Ai'(\zeta)&- \Ai'(\omega \zeta) \end{pmatrix}e^{-\frac{\pi i}{6}\sigma_3}, & \textrm{for $\zeta\in \Omega_4$,}
                 \end{array} \right.
   \end{equation}
where  $\omega=e^\frac{2\pi i}{3}$, $\Ai$ is the Airy function (cf. \cite[Chapter 9]{nist}),
$$M_\Ai=\sqrt{2\pi} e^{\frac 1 6\pi i}  \begin{pmatrix}
                                                                                 1 & 0 \\
                                                                                 0 & -i
                                                                                 \end{pmatrix}
$$ is a constant matrix and the regions $\Omega_i$ are indicated in Figure \ref{contour-for-model}. It is well-known that
$\Phi_{\Ai}$ solves the following RH problem; see \cite{dkmv2}.

\subsubsection*{RH problem for $\Phi_{\Ai}$}
\begin{enumerate}
\item[\rm (a)]  $\Phi_{\Ai}(\zeta)$ is analytic in
  $\mathbb{C}\setminus \{\cup^4_{j=1}\Sigma_j\cup\{0\}\}$, where the contours $\Sigma_j$, $j=1,2,3,4$  are illustrated in Figure \ref{contour-for-model}.

\item[\rm (b)]  $\Phi_{\Ai}$ satisfies the same jump condition \eqref{Psi-jump} as $\Psi$.

\item[\rm (c)] As $\zeta\to \infty$, we have
  \begin{equation}\label{Aipara-infty}
 \Phi_{\Ai}(\zeta) = \frac{1}{\sqrt{2}}\zeta^{-\frac{\sigma_3}{4}}\left[\begin{pmatrix}
 1 & i \\
 i & 1
 \end{pmatrix}+\frac{1}{48\zeta^{3/2}}
 \begin{pmatrix}
 -5 & 5i \\
 7i & -7
 \end{pmatrix}+O\left(\frac{1}{\zeta^3}\right)
 \right]e^{-\frac{2}{3}\zeta^{3/2}\sigma_3}.
   \end{equation}

 \end{enumerate}

The local parametrix $Q^{(-1)}(\zeta)$ is then constructed in terms of the Airy parametrix $\Phi_{\Ai}$ as follows:
\begin{equation}\label{Local-Airy}
Q^{(-1)}(\zeta)=
\begin{pmatrix}
1 & 0
\\
-i\frac{\tau_1}{s} & 1
\end{pmatrix}
s^{\frac{1}{4}\sigma_3}\Phi_{\Ai}(s(\zeta+1))e^{(s^{3/2}g(\zeta)+q(\zeta))\sigma_3}.
\end{equation}
With $Q^{(-1)}$ defined in \eqref{Local-Airy}, it is straightforward to check the jump condition stated in item (b) of the RH problem
for $Q^{(-1)}$ is satisfied. To show \eqref{Matching condition}, we see from \eqref{def:q}, \eqref{eq:c1asy}
\eqref{Out solution}, \eqref{Local-Airy} and \eqref{Aipara-infty} that, for $\zeta\in \partial U(-1,r)$ and large positive $s$,
\begin{align}\label{eq:mathcing}
&Q^{(-1)}(\zeta)Q^{(\infty)}(\zeta)^{-1}
\nonumber
\\
&=\begin{pmatrix}
1 & 0
\\
-i\frac{\tau_1}{s} & 1
\end{pmatrix}(\zeta+1)^{-\frac{1}{4}\sigma_3}\left[\begin{pmatrix}
 1 & i \\
 i & 1
 \end{pmatrix}+\frac{1}{48(s(\zeta+1))^{3/2}}
 \begin{pmatrix}
 -5 & 5i \\
 7i & -7
 \end{pmatrix}+O\left(\frac{1}{s^3}\right)\right]
 \nonumber
 \\
 &\qquad \times
 e^{q(\zeta)\sigma_3}\begin{pmatrix}
 1 & i \\
 i & 1
 \end{pmatrix}^{-1}(\zeta+1)^{\frac{1}{4}\sigma_3}
 \nonumber
 \\
 &=\begin{pmatrix}
1 & 0
\\
-i\frac{\tau_1}{s} & 1
\end{pmatrix}(\zeta+1)^{-\frac{1}{4}\sigma_3}\left[\begin{pmatrix}
 1 & i \\
 i & 1
 \end{pmatrix}+\frac{1}{48(s(\zeta+1))^{3/2}}
 \begin{pmatrix}
 -5 & 5i \\
 7i & -7
 \end{pmatrix}+O\left(\frac{1}{s^3}\right)\right]
 \nonumber
 \\
 &\qquad \times \left[I+\frac{c_1\sqrt{\zeta+1}}{\zeta}\frac{\sigma_3}{s}+O\left(\frac{1}{s^2}\right)\right]
 \begin{pmatrix}
 1 & i \\
 i & 1
 \end{pmatrix}^{-1}(\zeta+1)^{\frac{1}{4}\sigma_3}
 \nonumber
 \\
 &=I+\frac{J_1(\zeta)}{s}+\frac{J_2(\zeta)}{s^{3/2}}+O\left(\frac{1}{s^2}\right),
\end{align}
where
\begin{align}\label{def:J1}
J_1(\zeta)&=\frac{\tau_1}{\zeta}
\begin{pmatrix}
0 & -i \\
i & 0
\end{pmatrix},
\\
\label{def:J2}
J_2(\zeta)&=
\frac{1}{48}
\begin{pmatrix}
0 & \frac{5i}{(\zeta+1)^2} \\
\frac{7i}{\zeta+1} & 0
\end{pmatrix}.
\end{align}

%------------------------------------------------------------------------------------------------------------------------------

\subsubsection{Final transformation}
Our final transformation is defined by
\begin{equation}\label{def:D}
D(\zeta)=\left\{
       \begin{array}{ll}
        Q(\zeta)Q^{(\infty)}(\zeta)^{-1}, & \hbox{for $z \in \mathbb{C}\setminus U(-1,r) $,} \\
       Q(\zeta)Q^{(-1)}(\zeta)^{-1}, & \hbox{for $z \in U(-1,r)$.}
       \end{array}
     \right.
\end{equation}
\begin{figure}[h]
 \begin{center}
   \includegraphics[width=5.5cm]{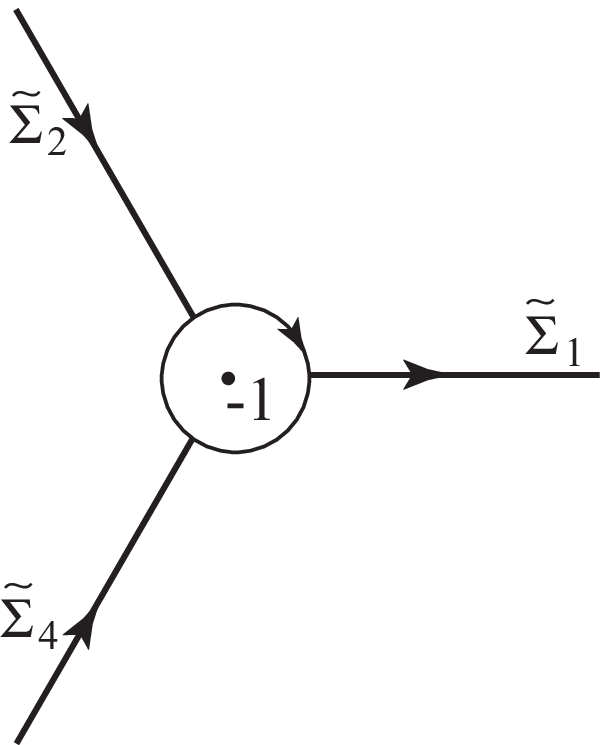} \end{center}
  \caption{Contour $\Sigma_D$ for the RH problem for $D$.}
 \label{contour-for-D}
\end{figure}
It is then easily seen that $D$ satisfies the following RH problem.

\subsubsection*{RH problem for $D$}
\begin{enumerate}
\item[\rm (a)]  $D(\zeta)$ is analytic in $\mathbb{C} \setminus \Sigma_{D}$, where the contour $\Sigma_D$ is shown
in Figure \ref{contour-for-D}.

\item[\rm (b)]  $D(\zeta)$  satisfies the jump condition
$$ D_{+}(\zeta) =D_{-}(\zeta)J_{D}(\zeta), \qquad \zeta\in \Sigma_D, $$
where
  \begin{equation}
                     J_{D}(\zeta)=\left\{
                                      \begin{array}{ll}
                                        Q^{(-1)}(\zeta)Q^{(\infty)}(\zeta)^{-1}, & \hbox{$\zeta \in \partial U(-1,r)$,} \\
                                        Q^{(\infty)}(\zeta) J_Q(\zeta) Q^{(\infty)}(\zeta)^{-1}, & \hbox{$\zeta \in \Sigma_{D} \setminus \partial U(-1,r)$,}
                                      \end{array}
\right.
\end{equation}
with $J_Q(\zeta)$ being defined in \eqref{Q-jump}.
\item[\rm (c)] As $\zeta \to \infty$,
\begin{equation}\label{D-large-zeta}
D(\zeta)=I+O(1/\zeta).
\end{equation}
\end{enumerate}

The RH problem for $D$ is equivalent to the following singular integral equation:
\begin{equation}\label{eq:integraloperatorD}
D(\zeta)=I+\frac{1}{2\pi i}\int_{\Sigma_D}D_-(w)\left(J_{D}(w)-I\right)\frac{dw}{w-\zeta}, \qquad \zeta \in\mathbb{C}\setminus \Sigma_D.
\end{equation}
Note that there exists some constant $c>0$ such that
\begin{equation}\label{eq:Dexpo}
J_{D}(\zeta)=I+O\left(e^{-cs^{3/2}}\right), \qquad \zeta \in \Sigma_{D} \setminus \partial U(-1,r)
\end{equation}
for large positive $s$. This, together with \eqref{eq:mathcing} and standard analysis (cf. \cite{deift,DZ93}), implies that $D(\zeta)$ admits a large $s$ expansion of the following form
\begin{equation}\label{D-expand}
D(\zeta)=I+\frac{D_1(\zeta)}{s}+\frac{D_2(\zeta)}{s^{3/2}}+O\left(s^{-2}\right),
\end{equation}
uniformly for $\zeta\in\mathbb{C}\setminus \Sigma_D$.

Furthermore, a combination of \eqref{D-expand} and the RH problem for $D$ shows that each $D_i(\zeta)$, $i=1,2$, satisfies the following RH problem.

\subsubsection*{RH problem for $D_i$}
\begin{enumerate}
\item[\rm (a)]  $D_i(\zeta)$ is analytic in $\mathbb{C} \setminus \partial U(-1,r)$.

\item[\rm (b)]  For $\zeta \in \partial U(-1,r)$, we have
  \begin{equation}
                    D_{i,+}(\zeta)-D_{i,-}(\zeta)=J_i(\zeta),
\end{equation}
where $J_i(\zeta)$, $i=1,2$, is given in \eqref{def:J1} and \eqref{def:J2}, respectively.
\item[\rm (c)] As $\zeta \to \infty$,
$$D_i(\zeta)=O(1/\zeta).$$
\end{enumerate}
By Cauchy theorem and the residue theorem, it is readily seen that
\begin{align}\label{eq:D1}
  D_1(\zeta) = \frac{1}{2\pi i} \oint_{\partial U(-1,r)}
  \frac{J_1(w)}{w-\zeta} dw
=\left\{
    \begin{array}{ll}
      \frac{\tau_1}{\zeta}
\begin{pmatrix}
0 & i \\
-i & 0
\end{pmatrix}, & \hbox{for $\zeta \in U(-1,r)$,} \\
      0, & \hbox{for $\zeta \in\mathbb{C}\setminus U(-1,r),$}
    \end{array}
  \right.
\end{align}
and
\begin{equation}\label{eq:D2}
  D_2(\zeta) = \frac{1}{2\pi i} \oint_{\partial U(-1,r)}
  \frac{J_2(w)}{w-\zeta} dw
= \frac{1}{48}
\begin{pmatrix}
0 & \frac{5i}{(\zeta+1)^2} \\
\frac{7i}{\zeta+1} & 0
\end{pmatrix}, \quad \hbox{for $\zeta \in\mathbb{C}\setminus U(-1,r)$.}
\end{equation}

We are now ready to prove Theorem \ref{thm:painleve solution}.

\subsection{Proof of Theorem \ref{thm:painleve solution}}
By Propositions \ref{thm:exsitence} and \ref{thm:laxpair}, it is immediate that there exists a family of solutions $b_1,\ldots, b_{2m+1}$ to the coupled Painlev\'{e} XXXIV system \eqref{eq:b-k} with the specified parameters \eqref{def:tauptilde}, which are also pole-free for real $s$.

To show that the function $b_1$ in \eqref{A-0} indeed has the large $s$ behavior \eqref{eq:b1asy}, by \eqref{a-1-b-1}, it suffices to derive the derive the large $s$ behavior of $a_1$. From \eqref{Out solution}, \eqref{def:D} and \eqref{eq:integraloperatorD}, it is readily seen that,
for large $\zeta$,
\begin{multline}\label{eq:Qasy}
Q(\zeta)=D(\zeta)Q^{(\infty)}(\zeta)
=\left(I+\frac{D_\infty}{\zeta}+O\left(\zeta^{-2}\right)\right) \biggl( I- \frac{1}{4\zeta} \sigma_3 +O\left(\zeta^{-2}\right) \biggr)
\zeta^{-\frac{1}{4} \sigma_3} \frac{I + i \sigma_1}{\sqrt{2}},
\end{multline}
where
\begin{equation}\label{def:Dinfty}
D_\infty=\frac{i}{2 \pi}\int_{\Sigma_D}D_-(w)\left(J_{D}(w)-I\right)dw.
\end{equation}
If we further take $s\to +\infty$, a combination of \eqref{def:Dinfty}, \eqref{D-expand}, \eqref{eq:Dexpo} and \eqref{def:J1} gives
\begin{equation}\label{D-inf-12}
(D_\infty)_{12}=O\left(s^{-2}\right).
\end{equation}
From  \eqref{eq:Qasy}, \eqref{def:Q-12}, \eqref{eq:c1asy} and the above estimate, we have
\begin{equation}\label{eq:a 1 asy}
a_1(s)=\frac{s^2}{4}+\frac{\tau_1}{\sqrt{s}}+O(1/s^{3/2}),\qquad s\to +\infty.
\end{equation}
This, together with \eqref{a-1-b-1}, gives us \eqref{eq:b1asy}.

This  completes the proof of Theorem \ref{thm:painleve solution}.
\qed

%--------------------------------------------------------------------------------------------------------------------------------------------------------------------------------------
\section{RH problem for orthogonal polynomials and the differential identities}
\label{sec:connectionRHP}
Recall that $\pi_n(x)$ is the monic polynomial of degree $n$ orthogonal with respect to the weight function $w(x)$ given in \eqref{eq:weight}, it is well-known that the $2 \times 2$ matrix-valued function
\begin{equation}\label{eq:Y}
Y(z)=Y\left(z;\lambda, \vec{t}\;\right)
=
\begin{pmatrix}
\pi_n(z) & \frac 1{2\pi i}\int_\mathbb{R} \frac{\pi_n(x) w(x)}{x-z}dx\\
-2\pi i \gamma_{n-1}^2 \;\pi_{n-1}(z)& - \gamma_{n-1}^2\;
\int_\mathbb{R} \frac{\pi_{n-1}(x) w(x)}{x-z}dx
\end{pmatrix}
\end{equation}
is the unique solution of the following RH problem (see \cite{fik}), where $\gamma_n$ depending on $\lambda$ is defined in \eqref{eq:orthogonal}.

\subsubsection*{RH problem for $Y$}
\begin{enumerate}
  \item [\rm(a)]  $Y(z)$ is analytic in $\mathbb{C} \setminus \mathbb{R}$.

  \item [\rm(b)]$Y(z)$  satisfies the jump condition
  \begin{equation}\label{eq:Yjump}
  Y_+(x)=Y_-(x)
  \begin{pmatrix}
  1 & w(x) \\
  0 & 1
  \end{pmatrix},\qquad x\in\mathbb{R}.
  \end{equation}

  \item [\rm(c)]  As $z\to \infty$, we have
  \begin{equation}\label{Y-infinity}
  Y(z)=\left (I+\frac {Y_{1}}{z}+O\left (\frac 1 {z^2}\right )\right)
 \begin{pmatrix}
 z^n & 0 \\
 0 & z^{-n}
 \end{pmatrix},
 \end{equation}
 where the matrix $Y_1$ is independent of $z$.
 \item [\rm(d)] $Y$ is bounded near $z=\lambda$.
\end{enumerate}

For later use, it is also worthwhile to point out that
\begin{equation}\label{eq:Y112}
(Y_1)_{12}=-\frac{1}{2 \pi i \gamma_n^2}, \qquad (Y_1)_{21}=- 2 \pi i \gamma_{n-1}^2;
\end{equation}
see \cite[Equation (3.11)]{dkmv2}.

In terms of the function $Y$ given in \eqref{eq:Y}, the correlation kernel \eqref{eq:correlation kernel} can be rewritten as
\begin{equation}\label{kernel-Y}
K_n\left(x,y; \lambda, \vec{t}\;\right)=\frac{\sqrt{w(x)w(y)}}{2\pi i(x-y)}\left(Y_+\left(y;\lambda, \vec{t}\;\right)^{-1}Y_+\left(x;\lambda, \vec{t}\;\right)\right)_{21}, \quad
x, y \in \mathbb{R}.
\end{equation}
The relation between the partition function $Z_n(\lambda)$ and $Y$ is more involved. We will derive two differential identities with respect to $\lambda$ in the following lemma, which are expressed in terms of the asymptotics of $Y$ near infinity and $z=\lambda$, respectively.

\begin{lem}\label{lem:differential identity}
Let $Z_n(\lambda)$ be the partition function defined in \eqref{def:partiationfunction}, then we have
\begin{align}
\label{diff identity}
    \frac d {d\lambda}\ln Z_{n}(\lambda&)=4n \lim_{z\to\infty}z(Y(z)z^{-n\sigma_3}-I)_{11}=4n(Y_1)_{11},
    \\
    \label{diff identity-2}
    \frac{ d^2}{d\lambda^2}\ln Z_{n}(\lambda)&=4n^2(\lambda^2-1)+\det\left(\frac{d}{d\lambda}H(\lambda)\cdot H(\lambda)^{-1}\right),
\end{align}
where $Y_1$ is the residue matrix appearing in the large $z$ asymptotics of $Y$ (see \eqref{Y-infinity}) and $H(\lambda)$ is defined by
\begin{equation}\label{def:H}
H(\lambda):=Y_+(\lambda)e^{-n\lambda^2\sigma_3}.
\end{equation}
\end{lem}
\begin{proof}
We start with the following expression for the partition function $Z_n(\lambda)$ in terms of $\gamma_k(\lambda)$ given in \eqref{eq:orthogonal}:
 \begin{equation}\label{Z:log H}
Z_{n}(\lambda)=n!\prod_{k=0}^{n-1}\gamma_{k}(\lambda)^{-2};
\end{equation}
see \cite{s}. Taking logarithmic derivative on both sides of the above formula gives us
\begin{equation}\label{Z:log derivative}
 \frac d{d\lambda}\ln Z_{n}(\lambda)=-2\sum_{k=0}^{n-1}\gamma_{k}(\lambda)^{-1}\frac d {d\lambda}\gamma_{k}(\lambda).
\end{equation}
To find the logarithmic derivative of $\gamma_k$, we see from \eqref{eq:orthogonal} and a change of variable $x\to x+\lambda$ that
\begin{equation}\label{leading coef-2}
\gamma_{k}(\lambda)^{-2}\delta_{j,k}=\int_{\mathbb{R}}\pi_j(x+\lambda)\pi_k(x+\lambda)w(x+\lambda)dx.
\end{equation}
By taking derivative with respect to $\lambda$ on both sides of \eqref{leading coef-2} with $j=k$, it follows that
\begin{equation*}
-2\gamma_k(\lambda)^{-3}\frac d {d\lambda}\gamma_{k}(\lambda)%=\int_{\mathbb{R}}(\pi_k(x+\lambda))^2 \frac {\partial} {\partial \lambda} w(x\lambda)dx
=-4 n \int_{\mathbb{R}}(x+\lambda)(\pi_k(x+\lambda))^2 w(x+\lambda)dx,
\end{equation*}
or, equivalently,
\begin{equation}\label{leading coef-d}
\gamma_{k}(\lambda)^{-1}\frac d {d\lambda}\gamma_{k}(\lambda) = 2n\gamma_{k}(\lambda)^2\int_{\mathbb{R}}x\pi_k(x)^2w(x)dx.
\end{equation}
This, together with  \eqref{Z:log derivative}, the Christoffel-Darboux formula for orthogonal polynomials and \eqref{eq:orthogonal}, implies that
\begin{align}\label{Z:integral}
\frac d{d\lambda}\ln Z_{n}(\lambda)
&=-4n\int_{\mathbb{R}} x\sum_{k=0}^{n-1}(\gamma_{k}(\lambda)\pi_k(x))^2w(x)dx
\nonumber
\\
&=-4n\gamma_{n-1}(\lambda)^2\int_{\mathbb{R}}x\left(\frac d {dx}\pi_n(x) \cdot \pi_{n-1}(x)-\pi_n(x) \cdot \frac d {dx}\pi_{n-1}(x)\right)w(x)dx
\nonumber
\\
&=-4n\gamma_{n-1}(\lambda)^2\int_{\mathbb{R}}x\frac d {dx}\pi_n(x) \cdot \pi_{n-1}(x) w(x)dx.
\end{align}
If we further set
$$ \pi_n(x)=x^n+ \mathfrak{p}_1(n;\lambda) x^{n-1}+\ldots $$
and expand $x\frac d{dx}\pi_n(x)$ in terms of $\pi_k$, it is readily seen that
\begin{equation}\label{Decomp.}
x\frac d{dx}\pi_n(x)=nx^n+(n-1)\mathfrak{p}_1(n;\lambda)x^{n-1}+\ldots = n\pi_n(x)-\mathfrak{p}_1(n;\lambda)\pi_{n-1}(x)+\ldots.
\end{equation}
Inserting the above formula into \eqref{Z:integral}, we obtain again from the orthogonality condition \eqref{eq:orthogonal} that
\begin{equation}\label{eq:Y111}
\frac d{d\lambda}\ln Z_{n}(\lambda)=4n\mathfrak{p}_1(n;\lambda)=4n \lim_{z\to\infty}z(Y(z)z^{-n\sigma_3}-I)_{11}=4n(Y_1)_{11},
\end{equation}
as required.

We next give the proof of the second differential identity \eqref{diff identity-2}. By taking derivative with respect to $\lambda$ on both sides of \eqref{leading coef-2} with $j=n$ and $k=n-1$, it is readily seen from the orthogonality condtion that
\begin{align*}
   0 & = \int_{\mathbb{R}}\left(\left( \frac d{d \lambda} \mathfrak{p}_1(n;\lambda)+n \right)(x+\lambda)^{n-1}+\ldots \right)\pi_{n-1}(x+\lambda)w(x+\lambda)dx
\\
& \qquad \qquad -
4n \int_{\mathbb{R}}(x+\lambda)\pi_{n-1}(x+\lambda)\pi_{n}(x+\lambda)w(x+\lambda)dx
\\
&= \int_{\mathbb{R}}\left(\left( \frac d{d \lambda} \mathfrak{p}_1(n;\lambda)+n \right)\pi_{n-1}(x+\lambda)+\ldots \right)\pi_{n-1}(x+\lambda)w(x+\lambda)dx
\\
& \qquad \qquad -4n
\int_{\mathbb{R}}(\pi_{n}(x+\lambda)+\ldots)\pi_{n}(x+\lambda)w(x+\lambda)dx
\\
&=\left( \frac d{d \lambda} \mathfrak{p}_1(n;\lambda)+n \right)\gamma_{n-1}(\lambda)^{-2}-4n\gamma_{n}(\lambda)^{-2}.
\end{align*}
Hence, a combination of \eqref{eq:Y111} and the above formula yields
\begin{equation}\label{eq: Y-1-diff}
  \frac{d}{d\lambda}(Y_1)_{11}=4n\left(\frac{\gamma_{n-1}(\lambda)}{\gamma_n(\lambda)}\right)^2-n=4n(Y_1)_{12}(Y_1)_{21}-n,
\end{equation}
where we have made use of \eqref{eq:Y112} in the last step.

To proceed, we define
\begin{equation}\label{def: Y-hat}
 \widetilde{Y}(z)=\widetilde{Y}(z;\lambda):=Y(z+\lambda)w(z+\lambda)^{\frac{1}{2}\sigma_3}.
\end{equation}
Thus, it is easily seen that $\widetilde{Y}$ solves the following RH problem.
\subsubsection*{RH problem for $\widetilde{Y}$}
\begin{enumerate}
  \item [\rm(a)]  $\widetilde{Y}(z)$ is analytic in $\mathbb{C} \setminus \mathbb{R}$.

  \item [\rm(b)] $\widetilde{Y}(z)$  satisfies the jump condition
  \begin{equation}\label{eq:tildeYjump}
  \widetilde{Y}_+(x)=\widetilde{Y}_-(x)
  \begin{pmatrix}
  1 & 1 \\
  0 & 1
  \end{pmatrix},\qquad x\in\mathbb{R}.
  \end{equation}

  \item [\rm(c)]  As $z\to \infty$, we have
  \begin{equation}\label{tildeY-infinity}
 \widetilde{Y}(z)=\left (I+\frac {\widetilde{Y}_{1}}{z}+O\left (\frac 1 {z^2}\right )\right)
 \begin{pmatrix}
 z^n & 0 \\
 0 & z^{-n}
 \end{pmatrix}e^{-n(z+\lambda )^2\sigma_3},
 \end{equation}
 where
 \begin{equation}
 \widetilde{Y}_{1}=Y_1
+
\begin{pmatrix}
-\frac{n}{2}t_1 & 0
\\
0 &\frac{n}{2}t_1
\end{pmatrix}
\end{equation}
with $Y_1$ being given in \eqref{Y-infinity}.
 \item [\rm(d)]
 As $z\to 0$, we have
  \begin{equation}\label{tildeY-zero}
  \widetilde{Y}\left(z;\lambda \right)=O(1)e^{-\frac{n}{2}\sum_{k=1}^{2m}t_kz^{-k}\sigma_3}.
  \end{equation}
\end{enumerate}
Since the jump matrix for $\widetilde {Y}$ is a constant matrix, we have that the function
$$\frac{\partial}{\partial \lambda}\widetilde {Y}(z;\lambda)\widetilde{Y}(z;\lambda)^{-1}$$
is meromorphic in $z$ with one possible singularity near the origin. By \eqref{tildeY-infinity}, it is readily seen that
\begin{equation}\label{eq: Y-hat-ODE}
\frac{\partial}{\partial \lambda}\widetilde{Y}(z;\lambda)=\left(-2n(z+\lambda)\sigma_3-4n\left(
                              \begin{array}{cc}
                                0 & -(Y_1)_{12} \\
                                (Y_1)_{21} & 0 \\
                              \end{array}
                            \right)\right)\widetilde{Y}(z;\lambda).
\end{equation}
By the definition \eqref{def:H}, we have $H(\lambda)=\widetilde {Y}_{+}(0)$.
From \eqref{eq: Y-hat-ODE}, we obtain that
\begin{equation}\label{eq:H-ODE}
\frac{d}{d\lambda}H(\lambda)=-2n\left(
                              \begin{array}{cc}
                                \lambda & -2(Y_1)_{12} \\
                                2(Y_1)_{21} & -\lambda \\
                              \end{array}
                            \right)H(\lambda).
\end{equation}
Hence, it follows from \eqref{diff identity}, \eqref{eq: Y-1-diff} and \eqref{eq:H-ODE} that
\begin{multline}
\frac{ d^2}{d\lambda^2}\ln Z_{n}(\lambda)
=4n\frac{d}{d\lambda}(Y_1)_{11}=16n^2(Y_1)_{12}(Y_1)_{21}-4n^2
\\=4n^2(\lambda^2-1)+\det\left(\frac{d}{d\lambda}H(\lambda)\cdot H(\lambda)^{-1}\right),
\end{multline}
which is \eqref{diff identity-2}.

This completes the proof of Lemma \ref{lem:differential identity}.
\end{proof}

%%%%%%%%%%%%%%%%%%%%%%%%%%%%%%%%%%%%%%%%%%%%%%%%%%%%%%%%%%%%%%%%%%%%%%%%%%%%
%-----------------------------------------------------------------------------------------------------------------------------------------------------------------------
\section{Asymptotic analysis of the RH problem for $Y$ with $1+cn^{-2/3}<\lambda<1+dn^{-1/3}$}
\label{sec:AARHY}

In this section, we perform the Deift-Zhou steepest descent analysis to the RH problem for $Y$ under \eqref{eq:tkscaling} and
 $$1+cn^{-2/3}<\lambda<1+dn^{-1/3},$$
where $c<0,d>0$ are arbitrarily fixed constants. The reason why we enlarge the range of $\lambda$ is explained at the end of Section \ref{sec:results}.

 \subsection{$Y\rightarrow T$: Normalization at $\infty$}
Define
\begin{equation}
\label{g}\widehat{g}(z)=\int_{-1}^{1}  \log(z-x) \varphi(x) dx, \qquad  z\in \mathbb{C}\setminus (-\infty, 1],
\end{equation}
where $\varphi(x)=\frac 2 {\pi } \sqrt{1-x^2}$ and the branch cut of the logarithm  is taken along the negative axis so that  $\arg (z-x)\in (-\pi , \pi)$.
We then introduce the first transformation $Y\rightarrow T$ to normalize the large $z$ behavior of $Y$:
 \begin{equation}\label{Y to T}
 T(z)= e^{-\frac 1 2 nl \sigma_3} Y(z) e^ {n \left (\frac 1 2 l
-\widehat{g}(z)\right )\sigma_3} e^{-\frac{n}{2}\sum_{k=1}^{2m}t_k(z-\lambda)^{-k}\sigma_3},\quad z\in
\mathbb{C} \setminus \mathbb{R}.
\end{equation}
with the constant
\begin{equation}\label{larg constant}
l:=-1-2\ln 2.
\end{equation}

With the $\phi$-function given in \eqref{phi}, it is easily seen that
the $\widehat{g}$-function and the $\phi$-function satisfy the following properties:
\begin{align}\label{eq:hatgdiff}
&\widehat{g}_{+}(x)-\widehat{g}_{-}(x)=\left\{
                    \begin{array}{ll}
                      2 \pi i, & \hbox{$x<-1$,} \\
                      2\phi_-(x)=-2\phi_+(x), & \hbox{$-1<x<1$,}
                    \end{array}
                  \right.
\\
&\widehat{g}(z)+\phi(z)-z^2-l/2=0, \qquad z\in\mathbb{C}. \label{eq:hatgphi}
\end{align}

Thus, it is readily seen that $T(z)$ defined in \eqref{Y to T} solves the following RH problem.
\subsubsection*{RH problem for T}
\begin{enumerate}
  \item[\rm(a)]
 $T(z)$ is analytic in $\mathbb{C} \setminus \mathbb{R}$.

  \item[\rm(b)] $T$ satisfies the jump condition
 \begin{equation}\label{Jump-T}
T_+(x)=T_-(x)\left\{
\begin{array}{ll}
\begin{pmatrix}
                                 1 & e^{-2n \phi(x)} \\
                                 0 & 1
\end{pmatrix},&   x\in(1,+\infty),\\
                             [.4cm]
 \begin{pmatrix}
                                 e^{2n\phi_+(x)} & 1 \\
                                 0 & e^{2n \phi_-(x)} \end{pmatrix},&   x\in(-1, 1), \\
                            [.4cm]
   \begin{pmatrix}
                                 1 &  e^{-2n \phi_+(x)} \\
                                 0 & 1 \end{pmatrix}, &  x\in(-\infty,-1).
\end{array}
\right .
\end{equation}
\item[\rm(c)] As $z\to \infty$, we have
  \begin{equation}\label{T-infinity}
  T(z)=I+O(1/z).
  \end{equation}
 \item[\rm(d)] As $z\to \lambda$, we have   \begin{equation}\label{T-zero}
  T(z)=O(1)e^{-\frac{n}{2}\sum_{k=1}^{2m}t_k(z-\lambda)^{-k}\sigma_3}.
  \end{equation}
\end{enumerate}

To show the jump condition \eqref{Jump-T}, we see from \eqref{eq:Yjump} and \eqref{Y to T} that
$$T_{+}(x)=T_-(x)
\begin{pmatrix}
e^{n(\widehat{g}_-(x)-\widehat{g}_+(x))} & e^{-2n(x^2+(l-\widehat{g}_+(x)-\widehat{g}_{-}(x))/2)}
\\
0 & e^{n(\widehat{g}_+(x)-\widehat{g}_-(x))}
\end{pmatrix}, \qquad x\in\mathbb{R}.$$
This, together with \eqref{eq:hatgdiff} and \eqref{eq:hatgphi}, implies \eqref{Jump-T}.

%------------------------------------------------------------------------------------------------------------------------------------------------------------------------------------------------------------------------
\subsection{$T \rightarrow S$: Contour deformation}
Since $\Re \phi_{\pm} (x) = 0$ for $x \in (-1,1)$ (see \eqref{phi}), the diagonal entries of the jump matrix in the RH problem for $T$ on $(-1, 1)$ are highly oscillatory for large $n$. We now deform the interval $[-1,\lambda]$ into a lens-shaped region (see Figure \ref{contour-for-S}), and introduce the transformation $T \to S$ below to remove the oscillations. Note that the lens opening depends on the parameter $\lambda$; see also \cite{xz2011} for a similar situation.
\begin{figure}[h]
 \begin{center}
   \includegraphics[width=7cm]{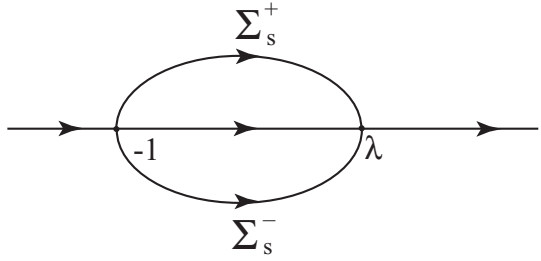}
 \end{center}
  \caption{Contour $\Sigma_S$ for the RH problem for $S$.}
 \label{contour-for-S}
\end{figure}

Define
\begin{equation}\label{T-S}
S(z)=\left \{
\begin{array}{ll}
  T(z), & \mbox{for $z$ outside the lens,}
  \\
  T(z) \begin{pmatrix}
                                 1 & 0 \\
                                  - e^{2n\phi(z)} & 1
\end{pmatrix}, & \mbox{for $z$ in the upper lens,}\\ [.4cm]
T(z) \begin{pmatrix}
                                 1 & 0 \\
                                  e^{2n\phi(z)} & 1
\end{pmatrix} , & \mbox{for $z$ in the lower lens. }
\end{array}\right.
\end{equation}
It is then straightforward to check that $S$ satisfies the following RH problem.
\subsubsection*{RH problem for $S$}
\begin{enumerate}
  \item[\rm(a)] $S(z)$ is analytic in
  $\mathbb{C} \setminus \Sigma_S$, where the contour $\Sigma_S$ is shown in Figure \ref{contour-for-S}.

  \item[\rm(b)] $S(z)$  satisfies the jump condition
 \begin{align}\label{Jump-S}
 S_+(z)&=S_-(z)J_{S}(z),
 \end{align}
where
\begin{align}
J_{S}(z)=S_-(z)\left\{
\begin{array}{ll}
\begin{pmatrix}
                                 1 & e^{-2n \phi(z)} \\
                                 0 & 1
\end{pmatrix},&   z \in(\max(\lambda,1), +\infty),
\\
                             [.4cm]
  \begin{pmatrix}
                                 0 & 1 \\
                                 -1 & 0 \end{pmatrix},&   z \in(-1, \min(\lambda, 1)),
                                 \\
                             [.4cm]
\begin{pmatrix}
                                 e^{2n\phi_+(z)} & 1 \\
                                 0 & e^{2n \phi_-(z)}\end{pmatrix}, &   z \in(\lambda, 1 ), \quad \textrm{if } \lambda < 1,  \\
                                                        [.4cm]
                             \begin{pmatrix}
                                 0 & e^{-2n\phi(z)} \\
                                 -e^{2n\phi(z)} & 0 \end{pmatrix}, &   z\in(1, \lambda ), \quad \textrm{if } \lambda > 1,  \\
                             [.4cm]
  \begin{pmatrix}
                                 1 &  e^{-2n \phi_+(z)} \\
                                 0 & 1
\end{pmatrix}, &  z\in(-\infty,-1),
\\
                             [.4cm]
\begin{pmatrix}
                                 1 &  0 \\
                                 e^{2n \phi(z)} & 1
\end{pmatrix}, & z \in \Sigma_S^+ \cup \Sigma_S^-.
\end{array}
\right .
\end{align}

  \item[\rm(c)] As $z\to \infty$, we have
  \begin{equation}
  S(z)=I+O(1/z).
  \end{equation}

 \item[\rm(d)] As $z \to \lambda$, we have
  \begin{equation}
  S(z)=O(1)e^{-\frac{n}{2}\sum_{k=1}^{2m}t_k(z-\lambda)^{-k}\sigma_3}.
  \end{equation}
\end{enumerate}

When $n$ is large, the jump matrix $J_S(z)$ tends to the identity matrix for $z$ bounded away from the interval $(-1, \lambda)$.
In what follows, we will construct both the outer parametrix and the local parametrices near endpoints to approximate $S$ for large $n$.
Particularly, the local parametrix near $z=1$ will be built in terms of the new model RH problem for $\Psi$ in Section \ref{sec:results}.

%-------------------------------------------------------------------------------------------------------------------------------------------------------------------------------------------------------------
\subsection{Outer parametrix} \label{sec:outer parametrix}
The outer parametrix $N$ solves an RH problem with a jump only along $(-1, \lambda)$.
\subsubsection*{RH problem for $N$}
\begin{enumerate}
\item[\rm(a)]  $N(z)$ is analytic in  $\mathbb{C} \setminus [-1,\lambda]$.
\item[\rm(b)]  $N$ satisfies the jump condition
      \begin{equation}\label{Jump-N} N_{+}(x)=N_{-}(x)\begin{pmatrix}
       0 & 1 \\
       -1 & 0
\end{pmatrix}, \qquad x\in  (-1,\lambda).
       \end{equation}
\item[\rm(c)]  As $z\to \infty$, we have
                \begin{equation}
                 N(z)= I+O(z^{-1}).
                 \end{equation}
  \end{enumerate}

The solution to the above RH problem is explicitly given by
   \begin{equation}\label{N}
  N(z)=
\begin{pmatrix}
    \frac {\beta_0(z) + \beta_0^{-1}(z)} {2}&\frac {\beta_0(z) - \beta_0^{-1}(z)} {2i} \\
    -\frac {\beta_0(z) - \beta_0^{-1}(z)} {2i} &\frac {\beta_0(z) + \beta_0^{-1}(z)} {2}
  \end{pmatrix},
  \end{equation}
 where $$\beta_0(z)=\left ( \frac {z-\lambda}{z+1} \right )^{1/ 4}$$ is analytic in $\mathbb{C}\setminus [-1, \lambda]$ and $\beta_0(z)\sim 1$ as $z\to\infty$.

\subsection{Local parametrix near $z=-1$}
Near $z=-1$, we seek a parametrix $P^{(-1)}(z)$ satisfying the following RH problem.

\subsubsection*{RH problem for $P^{(-1)}$}
\begin{enumerate}
\item[\rm (a)] $P^{(-1)}(z)$ is analytic in $\overline{U(-1,r)} \setminus
\Sigma_S$.

\item[\rm (b)] $P^{(-1)}(z)$ satisfies the jump condition
$$ P^{(-1)}_+(z)=P^{(-1)}_-(z)J_S(z),\qquad z \in U(-1,r) \cap \Sigma_S, $$
where the function $J_S(z)$ is given in \eqref{Jump-S}.

\item[\rm (c)] As $n \to \infty$, $P^{(-1)}(z)$ matches the outer parametrix $N(z)$ on the boundary of $U(-1,r)$, i.e.,
\begin{equation}\label{Matching condition:P-1}
P^{(-1)}(z)=
(I+O(1/n))N(z), \qquad  z \in \partial U(-1,r).
\end{equation}
\end{enumerate}
As what we did in Section \ref{subsubsec:Airypara}, the solution to RH problem for $P^{(-1)}$ can be constructed in terms of the Airy parametrix $\Phi_{\textrm{Ai}}$ \eqref{Airy-model-solution}, which we omit here.

%----------------------------------------------------------------------------------------------------------------------------------------------------------------------
\subsection{Local parametrix near $z=1$}
Recall that $\lambda \to 1$ as $n\to \infty$, thus, for $n$ large enough, both the points $z=1$ and $z=\lambda$ belong to an open dist $U(1,r)$ with $r$ small and fixed. We then intend to find a local parametrix $P^{(1)}(z)$ satisfying the following RH problem.

\subsubsection*{RH problem for $P^{(1)}$}
\begin{enumerate}
  \item[\rm(a)] $P^{(1)}(z)$  is analytic in $\overline{U(1,r)} \setminus \Sigma_{S}$.
  \item[\rm(b)] $P^{(1)}(z)$ satisfies the  jump condition
  \begin{equation}\label{PJS}
P^{(1)}_+(z)=P^{(1)}_-(z)J_{S}(z), \qquad z\in   U(1,r) \cap \Sigma_S.
\end{equation}
  \item[\rm(c)]  As $n\to \infty$, $P^{(1)}(z)$ matches $N(z)$ on the boundary of $U(1,r)$, i.e.,
\begin{equation}\label{mathcing condition P1}
P^{(1)}(z)= \left(I+ O\left (n^{-1/3}\right )\right)N(z), \qquad z\in \partial U(1,r).
 \end{equation}
 \item[\rm(d)]  $P^{(1)}(z)$ has the same local behavior as that of $S(z)$ near $z=\lambda$.
 \end{enumerate}

To solve the above RH problem, we note that the jump matrix in \eqref{PJS} can be reduced to a piecewise constant matrix by setting
\begin{equation}\label{eq:tildeP1}
  \widetilde{P}^{(1)}(z):= P^{(1)}(z)e^{-n\phi(z)\sigma_3}.
\end{equation}
It is easily seen that $\widetilde{P}^{(1)}(z)$ satisfies an RH problem as follows.

\subsubsection*{RH problem for $\widetilde{P}^{(1)}$}
\begin{enumerate}
  \item[\rm(a)] $\widetilde{P}^{(1)}(z)$  is analytic in $\overline{U(1,r)} \setminus  \Sigma_{S}$.
  \item[\rm(b)] $\widetilde{P}^{(1)}(z)$  satisfies the jump condition
  \begin{equation}\label{eq:jumpp1tilde}
 \widetilde{P}^{(1)}_+(z)=\widetilde{P}^{(1)}_-(z)\left\{ \begin{array}{ll}
  \begin{pmatrix}
    0 & 1 \\ -1 & 0
  \end{pmatrix}, & z \in (1 - r, \lambda), \\  [.4cm]
  \begin{pmatrix}
    1 & 1 \\ 0 & 1
  \end{pmatrix}, & z \in (\lambda, 1 + r),  \\
  [.4cm]
   \begin{pmatrix}
   1 &  0 \\
   1 & 1
   \end{pmatrix}, & z \in  U(1,r)\cap \{\Sigma_S^+ \cup \Sigma_S^-\}.
\end{array}
\right .
\end{equation}

  \item[\rm(c)] As $n\to\infty$, we have
\begin{equation}
\widetilde{P}^{(1)}(z)e^{n\phi(z)\sigma_3}= \left(I+ O\left (n^{-1/3}\right )\right)N(z), \qquad z\in \partial U(1,r).
 \end{equation}

 \item[\rm(d)] As $z\to \lambda$, we have
  \begin{equation} \label{eq:lambdap1tilde}
  \widetilde{P}^{(1)}(z)=O(1)e^{-\frac{n}{2}\sum_{k=1}^{2m}t_k(z-\lambda)^{-k}\sigma_3}.
  \end{equation}
\end{enumerate}

Recall the function $f$ defined in \eqref{def:f}, it is easily seen that, as $z \to 1$,
\begin{equation}\label{eq:fnear1}
f(z)=2(z-1)+O((z-1)^2),
\end{equation}
thus, it is analytic near $z=1$. Since $\lambda \to 1$ for large $n$, we have that $f(z)-f(\lambda)$ induces a conformal mapping from a neighborhood of $z=\lambda$ to that of $z=0$ for large $n$. By comparing \eqref{eq:jumpp1tilde} and \eqref{eq:lambdap1tilde} with \eqref{Psi-jump} and \eqref{Psi-origin}, respectively, we construct the local parametrix  ${P}^{(1)}$ with the aid of the model RH problem for $\Psi$ as follows:
\begin{equation}\label{parametrix}
P^{(1)}(z)=\widetilde{P}^{(1)}(z) e^{n\phi(z)\sigma_3} =E(z)\Psi\left (n^{2/3}(f(z)-f(\lambda));n^{2/3}f(\lambda), \vec{\tau}\right )e^{n\phi(z)\sigma_3},
\end{equation}
where
\begin{multline}\label{E}
E(z):=N(z)\frac{1}{\sqrt{2}}(I-i\sigma_1)
\left [n^{2/3}(f(z)-f(\lambda))\right ]^{\sigma_3/4}e^{\frac{1}{4}\pi i\sigma_3}
\\
\times
\begin{pmatrix}
1 & 0 \\
-a_1\left(n^{2/3}f(\lambda);\vec{\tau}\;\right) +\frac{n^{4/3}f(\lambda)^2}{4}& 1
\end{pmatrix}.
\end{multline}
In \eqref{E}, the $\frac{1}{4}$-root takes the principal branch. Thus,
on account of \eqref{eq:fnear1}, one has
\begin{equation}
  \biggl( f(x)-f(\lambda) \biggr)_+^{1/4} = e^{\pi i /2} \biggl( f(x)-f(\lambda) \biggr)_-^{1/4},
\end{equation}
for $x<\lambda$ and in a small neighborhood of $\lambda$. This, together with the RH problem for $N$ given in Section \ref{sec:outer parametrix}, implies that the pre-factor $E$ in \eqref{parametrix} is analytic in a neighborhood of $z = \lambda$, hence also in $U(1,r)$ for large $n$.

With $P^{(1)}$ defined in \eqref{parametrix}, it is easily seen that items (a) and (b) in the RH problem for $P^{(1)}$ have been fulfilled. It then remains to check the matching condition \eqref{mathcing condition P1} on the boundary of the circle and its local behavior near $z = \lambda$.

For the local behavior near $z = \lambda$, by \eqref{Psi-origin}, \eqref{eq:tildeP1}, \eqref{eq:lambdap1tilde} and \eqref{parametrix},  it suffices to show that the function
\begin{equation}\label{def:p}
p(z;n):=\sum_{j=1}^{2m}\tau_jn^{-2j/3}(f(z)-f(\lambda))^{-j}-\frac{n}{2}\sum_{k=1}^{2m}t_k(z-\lambda)^{-k}
\end{equation}
is bounded as $z\to \lambda$. In view of the expansion \eqref{def:cjk},
we have
\begin{equation}\label{eq: p-expand}
p(z;n)=\sum_{k=1}^{2m}\left(\sum_{j=k}^{2m}c_{jk}\tau_j n^{-2j/3}-\frac{n}{2}t_k\right)(z-\lambda)^{-k}+\sum_{j=1}^{2m}c_{j0}\tau_j n^{-2j/3}+O(z-\lambda).
\end{equation}
Recall the scaling of $\vec{t}=(t_1,t_2,\ldots,t_{2m})$ given in \eqref{eq:tkscaling}, we have that the coefficients of $(z-\lambda)^{-k}$, $k=1,...,2m$ in the above formula all vanish. Thus, $p(z;n)$ is analytic at $z=\lambda$, which also implies its boundedness near $z=\lambda$. Moreover, we have the estimate
\begin{equation}\label{def:p-lambda}
p(\lambda;n)=\sum_{j=1}^{2m}c_{j0}\tau_jn^{-2j/3}=O(n^{-2/3})
\end{equation}
for large $n$.

We proceed to verify the matching condition \eqref{mathcing condition P1} on $\partial U(1,r)$. The discussion is divided into two cases depending on the range of $\lambda$, namely, $cn^{-2/3}< \lambda-1 < c_0 n^{-2/3}$ and $c_0n^{-2/3}< \lambda-1 < d n^{-1/3}$, where $c_0$ is an arbitrary positive and big enough constant.
In the former case, we have that $n^{2/3}f(\lambda)$ is bounded.   We then obtain from \eqref{parametrix}, \eqref{E} and \eqref{Psi-infty} that, for $z\in \partial U(1,r)$ and large $n$,
\begin{align}\label{eq:matchingcheck1}
P^{(1)}(z)&N(z)^{-1}=N(z)\frac{1}{\sqrt{2}}(I-i\sigma_1)
\left [n^{2/3}(f(z)-f(\lambda))\right ]^{\sigma_3/4}e^{\frac{1}{4}\pi i\sigma_3}
\nonumber
\\
&\times \begin{pmatrix}
1 & 0 \\
\frac{n^{4/3}f(\lambda)^2}{4}& 1
\end{pmatrix}\left(I+
\frac{
\Psi_{1}(n^{2/3}f(\lambda))}{n^{2/3}(f(z)-f(\lambda))} +
O\left(n^{-4/3}\right)\right)
\nonumber
\\
&\times e^{-\frac{1}{4}\pi i\sigma_3}\left [n^{2/3}(f(z)-f(\lambda))\right ]^{-\sigma_3/4}\frac{1}{\sqrt{2}}(I+i\sigma_1)e^{F_n(z)\sigma_3}N(z)^{-1},
\end{align}
where
\begin{align}
F_n(z)&:= n\phi(z)- \theta\left( n^{2/3}(f(z)-f(\lambda));n^{2/3}f(\lambda) \right) \nonumber
\\
&=\frac{2}{3}n f(z)^{\frac{3}{2}}-\frac{2}{3}n(f(z)-f(\lambda))^{\frac{3}{2}}-nf(\lambda)(f(z)-f(\lambda))^{\frac{1}{2}}. \label{F-n}
\end{align}
Note that $|f(z)|$ is uniformly bounded below from $0$ for $z\in \partial U(1,r)$ and $f(\lambda)\to 0$ as $n$ goes to infinity, it then follows from a direct calculation and the Taylor expansion that
\begin{align}
F_n(z)&=\frac{nf(\lambda)^2}{4(f(z)-f(\lambda))^{\frac12}}\left(-\frac{8}{3}\frac{f(z)^2}{f(\lambda)^2}+\frac{4}{3}\frac{f(z)}{f(\lambda)}+\frac43
+\frac{8}{3}\frac{f(z)^2}{f(\lambda)^2}\left(1-\frac{f(\lambda)}{f(z)}\right)^{\frac12}\right)
\nonumber
\\
&=\frac{nf(\lambda)^2}{4(f(z)-f(\lambda))^{\frac12}}\left(1+O\left(\frac{f(\lambda)}{f(z)}\right)\right).
\end{align}
Since $n^{2/3}f(\lambda)$ is bounded, it follows that
\begin{equation}
F_n(z)=O(n^{-1/3})
\end{equation}
for $z\in \partial U(1,r)$ and large $n$. Inserting the above formula into \eqref{eq:matchingcheck1}, it follows that
\begin{align}
 P^{(1)}(z)N(z)^{-1}
& =N(z)\big(I+\frac{(n^{4/3}f(\lambda)^2/4-a_1(n^{2/3}f(\lambda)))(\sigma_3-i\sigma_1)}{2n^{1/3}\sqrt{f(z)-f(\lambda)}}
\nonumber
\\
& \qquad +O (n^{-2/3})\big) (I+O(n^{-1/3}))N(z)^{-1}
\nonumber
\\
&=I+O(n^{-1/3}),
\nonumber
\end{align}
as required.

If $c_0 n^{-2/3}<\lambda-1<dn^{-1/3}$, the function $n^{2/3}f(\lambda)$ might be unbounded. The expansion \eqref{eq:matchingcheck1} is not valid anymore, since
the asymptotics of $\Psi(\zeta; s , \vec{\tau})$ in \eqref{Psi-infty} does not hold for $s$ large. Based on the asymptotic analysis of the RH problem for $\Psi$ performed in Section \ref{sec-Psi-large-S}, however, we could derive the other asymptotic formula of $\Psi(\zeta; s , \vec{\tau} )$ for both $s$ and $\zeta$ large, as stated in the next lemma. This expansion will then be used to verify the matching condition in the second case.

\begin{lem}\label{lem:asypsinew}
  We have
  \begin{align} \label{eq:Psi-large s-1}
    \Psi(\zeta; s , \vec{\tau} ) & = \left(
                  \begin{array}{cc}
                    1 & 0 \\
                    a_1(s ; \vec{\tau})-\frac{s^2}{4} & 1 \\
                  \end{array}
                \right)e^{-\frac{1}{4}\pi i\sigma_3}\zeta^{-\frac{1}{4}\sigma_3}\frac{I+i\sigma_1}{\sqrt{2}} \nonumber \\
                & \qquad \qquad \times \left(I+O(1/\sqrt{s\zeta})+O(s/\zeta)\right)\exp\left(-\frac{2}{3}(\zeta+s)^{3/2}\sigma_3 \right),
  \end{align}
  as $s \to \infty$ and $\zeta/s\to \infty$.
\end{lem}

\begin{proof}
Tracing back the transformations $\Psi\to U\to W \to Q \to D$ in \eqref{psi to U}, \eqref{U to W}, \eqref{W to Q} and \eqref{def:D}, it follows that,
if $|\zeta+s|>\delta$,
\begin{align} \label{eq:Psi-large s-0}
  \Psi(\zeta)=&\left(
              \begin{array}{cc}
                1 & 0 \\
                a_1-\frac{s^2}{4}-\frac{c_1}{\sqrt{s}} & 1 \\
              \end{array}
            \right)e^{-\frac{1}{4}\pi i\sigma_3}s^{-\frac{1}{4}\sigma_3}D(\zeta/s)\left(\frac{\zeta}{s}+1\right)^{-\frac{1}{4}\sigma_3}
            \frac{I+i\sigma_1}{\sqrt{2}} \nonumber \\
            & \qquad \times \exp\left(-(s^{3/2}g(\zeta/s)+q(\zeta/s))\sigma_3\right).
\end{align}
As $\zeta/s \to \infty$, it is readily seen from \eqref{def:q}, \eqref{D-large-zeta} and \eqref{eq:integraloperatorD} that
\begin{equation*}
  q(\zeta/s)= \frac{c_1}{(s \zeta)^{1/2}} + O\left(\frac{s^{1/2}}{\zeta^{3/2}}\right) ,
\end{equation*}
and
\begin{equation*}\label{eq:D-est}
D(\zeta/s)=I+ \frac{s}{\zeta} \cdot D_{\infty} + O\left(\frac{s^2}{\zeta^2}\right),
\end{equation*}
where the matrix $D_{\infty}$ is given in \eqref{def:Dinfty}. Also note that
\begin{equation*}
  \left(\frac{\zeta}{s}+1\right)^{-\frac{1}{4}\sigma_3}=\left(\frac{\zeta}{s}\right)^{-\frac{1}{4}\sigma_3} \left( I - \frac{s}{4\zeta} \sigma_3 + O\left(\frac{s^2}{\zeta^2}\right) \right) , \qquad \textrm{as } \zeta/s \to \infty,
\end{equation*}
and
\begin{equation*}
s^{3/2}g(\zeta/s) = \frac{2}{3}(\zeta+s)^{3/2};
\end{equation*}
see the definition of $g(\zeta)$ in \eqref{def:g}. Inserting the above formulas into \eqref{eq:Psi-large s-0} gives us \eqref{eq:Psi-large s-1}. Here, we have also make use of the fact that $s (D_\infty)_{12}$ is uniformly bounded for all $c_0< s < +\infty$; see \eqref{D-inf-12}.

This completes the proof of Lemma \ref{lem:asypsinew}.
\end{proof}

We now return to checking the matching condition for $c_0 n^{-2/3}<\lambda-1<dn^{-1/3}$. To this end, note that, for large $n$,
\begin{equation*}
  \zeta:=n^{2/3}(f(z)-f(\lambda)) = O(n^{2/3}), \qquad  z\in \partial U(1,r),
\end{equation*}
 \begin{equation*}
 n^{2/3}f(\lambda)=O(n^{1/3}),
\end{equation*}
and
\begin{equation*}
 \frac{2}{3}(\zeta+n^{2/3}f(\lambda))^{3/2}= \frac{2}{3}nf(z)^{3/2}=n\phi(z).
  \end{equation*}
Thus, we use the asymptotic formula \eqref{eq:Psi-large s-1} in the large $n$ expansion of \eqref{parametrix}, and it follows from a straightforward calculation that
\begin{equation*}
P^{(1)}(z) N(z)^{-1}= I+O(n^{-1/3}), \qquad  z\in \partial U(1,r).
\end{equation*}
As a consequence, we conclude that $P^{(1)}(z)$ in \eqref{parametrix} satisfies the matching condition \eqref{mathcing condition P1} for all $c n^{-2/3}<\lambda-1<dn^{-1/3}$.

%%%%%%%%%%%%%%%%%%%%%%%%%%%%%%%%%%%%%%%%%%%%%%%%%%%%%%%%%%%%%%%%%%%%%%%%%%%%%%%%%%%%%%%%%%%%%%%%%%%%%%%%%%%%%%%%%%%%%%%%%%%%%%%%%%%%%%%
\subsection{Final transformation}
The final transformation is defined by
\begin{equation}\label{S-R}
R(z)=\left\{ \begin{array}{ll}
                S(z)N(z)^{-1}, & \textrm{for $z\in \mathbb{C} \setminus \left \{ U(-1,r)\cup U(1,r)\cup \Sigma_S \right \}$,} \\
               S(z) P^{(-1)}(z)^{-1}, & \textrm{for $z\in   U(-1,r) \setminus \Sigma_{S}$,}  \\
               S(z)  P^{(1)}(z)^{-1}, & \textrm{for $z\in   U(1,r) \setminus \Sigma_{S}$.}
             \end{array}\right .
\end{equation}
Then, it is readily seen that $R$ satisfies the following RH problem.
\begin{figure}[h]
 \begin{center}
   \includegraphics[width=7.5cm]{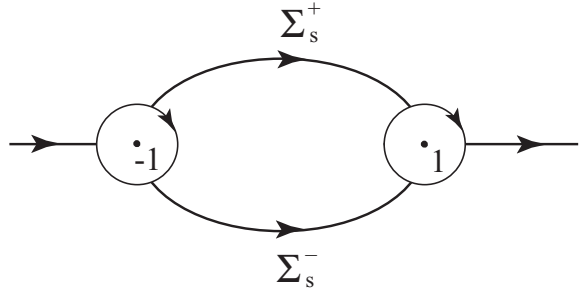} \end{center}
  \caption{The contour $\Sigma_R$ for the RH problem for $R$.}
 \label{contour-for-R}
\end{figure}

\subsubsection*{RH problem for $R$}
\begin{enumerate}
\item[\rm (a)]  $R(z)$ is analytic in $\mathbb{C} \setminus \Sigma_{R}$, where the contour $\Sigma_R$ is shown in Figure \ref{contour-for-R}.

\item[\rm (b)]  $R$  satisfies the jump condition
$$ R_{+}(z) = R_{-}(z)J_{R}(z), \qquad z\in \Sigma_R,$$
where
  \begin{equation}
                     J_{R}(z)=\left\{
                                      \begin{array}{ll}
                                      P^{(1)}(z)N(z)^{-1}, & \hbox{for $z \in \partial U(1,r)$,} \\
                                        P^{(-1)}(z)N(z)^{-1}, & \hbox{for $z \in \partial U(-1,r)$,} \\
                                        N(z) J_S(z) N(z)^{-1}, & \hbox{for $z \in \Sigma_{R} \setminus \{ \partial U(-1,r) \cup \partial U(1,r) \}$.}
                                      \end{array}
\right.
  \end{equation}
\item[\rm (c)] As $z \to \infty$,
$$R(z)=I+O(1/z).$$
\end{enumerate}

For $z \in \Sigma_{R} \setminus \{ \partial U(-1,r) \cup \partial U(1,r) \}$, it is easily seen that
\begin{equation}\label{eq:Rexpo}
J_{R}(z)=I+O\left(e^{-\tilde {c} n}\right),
\end{equation}
for some $\tilde {c}>0$ and large $n$. Thus, in view of \eqref{Matching condition:P-1}, \eqref{mathcing condition P1} and the above estimate, it follows again from the
standard analysis for the small norm RH problem that
\begin{equation}\label{R-asymptotic}
 R(z)=I+O(n^{-1/3}), \qquad n\to \infty,
\end{equation}
where the error bound is uniformly valid for $z \in \mathbb{C} \setminus \Sigma_{R}$ and  $1+cn^{-2/3}<\lambda<1+dn^{-1/3}$
with some constants $c<0$ and $d>0$.

%-------------------------------------------------------------------------------------------------------------------------------------------------------------------------
\section{Proof of  Theorem \ref{thm:Correlation kernel-Asy} and asymptotics of $\frac{d^2}{d\lambda^2}\ln Z_n(\lambda)$}
\label{sec:pfthmkernel}

As a consequence of \eqref{R-asymptotic}, we are now able to prove Theorem \ref{thm:Correlation kernel-Asy} and derive the asymptotics of
the logarithmic derivative of the partition function $\frac{d^2}{d\lambda^2}\ln Z_n(\lambda)$ in this section. We begin with the proof of Theorem \ref{thm:Correlation kernel-Asy}.

\subsection{Proof of Theorem \ref{thm:Correlation kernel-Asy}}

\label{sec:pfasykernel}
From \eqref{kernel-Y}, it follows that
\begin{equation}\label{eq:kerY}
K_n\left(x,y; \lambda, \vec{t}\;\right)=\frac{\sqrt{w(x)w(y)}}{2\pi i(x-y)}
\begin{pmatrix}
0 & 1
\end{pmatrix}
Y_+\left(y;\lambda, \vec{t}\;\right)^{-1}Y_+\left(x;\lambda, \vec{t}\;\right)
\begin{pmatrix}
1
\\
0
\end{pmatrix}.
\end{equation}
The large $n$ approximation of $Y_+(x)$ can be obtained by tracing back the sequence of transformations $Y\to T\to S\to R$, which gives us
\begin{align}\label{Y-Psi-1}
Y_+(x)
=\left\{
         \begin{array}{ll}
           e^{\frac{1}{2}nl\sigma_3}R(x)E(x)\Psi_+w(x)^{-\sigma_3/2}, & \hbox{$\lambda<x<1-r$,}
        \\[.4cm]
          e^{\frac{1}{2}nl\sigma_3}R(x)E(x)\Psi_+\begin{pmatrix}
                                                                            1 & 0 \\
                                                                            1 & 1
\end{pmatrix}
w(x)^{-\sigma_3/2}, & \hbox{$1-r <x < \lambda$,}
         \end{array}
       \right.
\end{align}
where $\Psi_+=\Psi_+(n^{2/3}(f(x)-f(\lambda));n^{2/3}f(\lambda),\vec{\tau})$.
Now we fix $u,v>0$ and take
\begin{equation}\label{eq:xy}
x=\lambda+\frac{u}{2n^{2/3}},\qquad   y=\lambda+\frac{v}{2n^{2/3}}.
\end{equation}
It is then easily seen from \eqref{eq:kerY} and \eqref{Y-Psi-1} that
\begin{multline}\label{eq:kuvpositive}
K_n\left(x,y; \lambda, \vec{t}\;\right)
=\frac{1}{2 \pi i (x-y)}\begin{pmatrix}
0 & 1
\end{pmatrix}
\Psi_+(n^{2/3}(f(y)-f(\lambda));n^{2/3}f(\lambda),\vec{\tau})^{-1}\\
\times
E(y)^{-1}R(y)^{-1}
 R(x)E(x)\Psi_+(n^{2/3}(f(x)-f(\lambda));n^{2/3}f(\lambda),\vec{\tau})
\begin{pmatrix}
1
\\
0
\end{pmatrix}.
\end{multline}
As $n\to \infty$, it is immediate from \eqref{eq:lambdascalingparti}, \eqref{eq:fnear1} and \eqref{eq:xy} that
\begin{equation}\label{eq:limitxy}
n^{2/3}f(\lambda)  \to s, \qquad
n^{2/3}(f(x)-f(\lambda))  \to u, \qquad n^{2/3}(f(y)-f(\lambda))  \to v.
\end{equation}
Furthermore, since both $R$ and $E$ are analytic near $z=\lambda$, we have
\begin{equation}\label{RinverseR}
R(y)^{-1}R(x)=I+O\left(\frac{x-y}{n^{1/3}}\right) =
I+O\left(\frac{u-v}{n}\right),
\end{equation}
and in view of \eqref{E} we see that $E(x)=O\left(n^{1/6}\right)$,
$E(y)=O\left(n^{1/6}\right)$ and
\begin{equation}\label{EinverseE}
E(y)^{-1}E(x) = I+O(n^{-1/3}),
\end{equation}
as $n\to\infty$. As a consequence, one has
\begin{equation}\label{ER}
E(y)^{-1}R^{-1}(y)R(x)E(x)=I+O(n^{-1/3}).
\end{equation}
Inserting \eqref{eq:limitxy} and \eqref{ER} into \eqref{eq:kuvpositive} gives us
\eqref{Cor kernel-asy} for $u,v>0$.

The case where $u$ and/or $v$ are negative can be proved in a similar manner. We do not give details here.

This completes the proof of Theorem \ref{thm:Correlation kernel-Asy}.
\qed

%-----------------------------------------------------------------------------------------------------------------------------------------------------------------------
\subsection{Asymptotics of $\frac{d^2}{d\lambda^2}\ln Z_n(\lambda)$}
It is the aim of this subsection to prove the following lemma concerning the asymptotics
of the logarithmic derivative of the partition function, which will be the starting point in our proof of Theorem \ref{thm:partition asymptotics}.
The proof is based on the connection between $\frac{d^2}{d\lambda^2}\ln Z_n(\lambda)$ and the RH problem for $Y$ established in \eqref{diff identity-2}.

\begin{lem}\label{thm:Log-Partition function-asy}
Under the condition \eqref{eq:tkscaling}, we have, as $n\to \infty$,
\begin{equation}\label{log-Z-n-asy}
    \frac{d^2}{d\lambda^2}\ln Z_n(\lambda)=-\left(n^{2/3} \frac{d f(\lambda)}{d \lambda}\right)^2 \biggl( a'_1(n^{2/3}f(\lambda)) -\frac{n^{2/3}f(\lambda)}{2} + O(n^{-2/3}) \biggr),
\end{equation}
where $f$ is defined in \eqref{def:f}, $a_1(s)=a_1(s;\vec{\tau})$ is given in \eqref{a1:def}, $a_1'(s) = \frac{da_1(s)}{ds}$ and the error bound is uniform for $1+cn^{-2/3}<\lambda<1+dn^{-1/3}$ with any choice of  constants $c<0$ and $d>0$.
\end{lem}

\begin{proof}
In view of \eqref{diff identity-2}, our main task is to estimate $\det\left(\frac{d}{d\lambda}H(\lambda)\cdot H(\lambda)^{-1}\right)$ for large $n$.
We first observe from \eqref{def:H} and \eqref{Y-Psi-1} that
\begin{equation}\label{eq:H-lambda}
H(\lambda)=e^{\frac{1}{2}nl\sigma_3}R(\lambda) E(\lambda)\Psi_0( n^{2/3}f(\lambda))e^{-p(\lambda;n)\sigma_3},
\end{equation}
where $\Psi_0(s)$ and $p(\lambda;n)$ are defined in \eqref{Psi-origin} and \eqref{def:p-lambda}, respectively.
Moreover, from the definitions of $N(z)$ in \eqref{N} and $E(z)$ in \eqref{E}, we have
\begin{align}
E(\lambda) & = \frac{I-i\sigma_1}{\sqrt{2}} \lim_{z \to \lambda} \biggl[  \left ( \frac {z-\lambda}{z+1} \right )^{- \frac{1}{4}\sigma_3} \left( n^{2/3}(f(z)-f(\lambda))\right)^{\sigma_3/4} \biggr]
  \, e^{\frac{1}{4}\pi i\sigma_3} \nonumber \\
  & \qquad \qquad \times \begin{pmatrix}
1 & 0 \\
-a_1\left(n^{2/3}f(\lambda)\right) + \frac{1}{4}n^{4/3}f(\lambda)^2 & 1 \\
\end{pmatrix} \nonumber \\
& = \frac{I-i\sigma_1}{\sqrt{2}} \biggl(\frac{d f(\lambda)}{d \lambda} (\lambda+1)\biggr)^{\frac{1}{4}\sigma_3}e^{\frac{1}{4}\pi i\sigma_3}n^{\frac{1}{6}\sigma_3}  
  \begin{pmatrix}
1 & 0 \\
-a_1(n^{2/3}f(\lambda)) + \frac{1}{4}n^{4/3}f(\lambda)^2 & 1 \\
\end{pmatrix}.
\end{align}

To estimate $\frac{d}{d\lambda}H(\lambda)\cdot H(\lambda)^{-1}$, it is convenient to decompose the function $H(\lambda)$ into three terms as follows:
\begin{equation}\label{eq:H-lambda-1}
H(\lambda)=H_0(\lambda)H_1(\lambda)e^{-p(\lambda;n)\sigma_3}.
\end{equation}
where
\begin{align}\label{def:H-0}
H_0(\lambda)&:= e^{\frac{1}{2}nl\sigma_3}R(\lambda)\frac{I-i\sigma_1}{\sqrt{2}} \biggl( \frac{d f(\lambda)}{d \lambda} (\lambda+1) \biggr)^{\frac{1}{4}\sigma_3}e^{\frac{1}{4}\pi i\sigma_3}n^{\frac{1}{6}\sigma_3},
\\
\label{def:H-1}
H_1(\lambda)&:= \left(
                                                                                          \begin{array}{cc}
                                                                                            1 & 0 \\
                                                                                            -a_1( n^{2/3}f(\lambda))+\frac{1}{4} n^{4/3}f(\lambda)^2 & 1 \\
                                                                                          \end{array}
                                                                                        \right)\Psi_0( n^{2/3}f(\lambda)).
\end{align}
Then, we have from \eqref{eq:H-lambda-1} that
\begin{align}
  \frac{d}{d\lambda}H(\lambda) \cdot H(\lambda)^{-1}& =\frac{d}{d\lambda}H_0(\lambda)\cdot H_0(\lambda)^{-1}+H_0(\lambda)\cdot \frac{d}{d\lambda}H_1(\lambda)\cdot H_1(\lambda)^{-1}H_0(\lambda)^{-1} \nonumber \\
   & \qquad -\frac{d}{d\lambda}p(\lambda;n) \cdot H_0(\lambda)H_1(\lambda)\sigma_3H_1(\lambda)^{-1}H_0(\lambda)^{-1}.
   \label{def:H-diff}
\end{align}
Since $H_0(\lambda)$ is non-singular, we obtain
\begin{align}
   &\det\left( \frac{d}{d\lambda}H(\lambda) H(\lambda)^{-1}\right)  = \det\left( H_0(\lambda)^{-1}  \frac{d}{d\lambda}H(\lambda) H(\lambda)^{-1} H_0(\lambda)\right)   \nonumber \\
  &  =\det \big( H_0(\lambda)^{-1} \frac{d}{d\lambda}H_0(\lambda) 
+\frac{d}{d\lambda}H_1(\lambda) H_1(\lambda)^{-1} -\frac{d}{d\lambda}p(\lambda;n) H_1(\lambda)\sigma_3H_1(\lambda)^{-1} \big). \label{def:H-diff-1}
\end{align}

We next estimate the above expression term by term. From \eqref{def:H-0}, it follows that
\begin{align*}
   & H_0(\lambda)^{-1} \cdot \frac{d}{d\lambda}H_0(\lambda)= n^{-\frac{1}{6}\sigma_3} e^{-\frac{1}{4}\pi i\sigma_3}\left( \frac{d f(\lambda)}{d \lambda} (\lambda+1) \right)^{-\frac{1}{4}\sigma_3} \frac{I+i\sigma_1}{\sqrt{2}} R(\lambda)^{-1} \\
  & \qquad \cdot   \frac{d}{d\lambda}R(\lambda)  \frac{I-i\sigma_1}{\sqrt{2}} \left( \frac{d f(\lambda)}{d \lambda} (\lambda+1) \right)^{\frac{1}{4}\sigma_3} e^{\frac{1}{4}\pi i\sigma_3} n^{\frac{1}{6}\sigma_3} + \frac{1}{4} \left( \frac{\frac{d^2 f(\lambda)}{d \lambda^2} (\lambda+1) +  \frac{d f(\lambda)}{d \lambda}}{\frac{d f(\lambda)}{d \lambda} (\lambda+1)} \right)\sigma_3.
\end{align*}
On account of \eqref{R-asymptotic} and the fact that both $\frac{d f(\lambda)}{d \lambda}$ and $\frac{d^2 f(\lambda)}{d \lambda^2}$ are bounded for large $n$, we have
\begin{equation}\label{def:H-0-diff}
H_0(\lambda)^{-1}\cdot \frac{d}{d\lambda} H_0(\lambda) =
\begin{pmatrix}
  O(1) & O(n^{-2/3}) \\ O(1) & O(1)
\end{pmatrix},  \qquad \textrm{as } n \to \infty.
\end{equation}
To compute $\frac{d}{d\lambda}H_1(\lambda) \cdot H_1(\lambda)^{-1}$, we first use \eqref{Psi-origin} and \eqref{Lax pair-2} to obtain
\begin{align*}
  B(\zeta;s) = \frac{\partial\Psi}{\partial s} \Psi^{-1} = \left( \frac{d\Psi_0(s)}{ds}\left[I+O(\zeta)\right] + \Psi_0(s) O(\zeta) \right) \left[I+O(\zeta)\right]^{-1} \Psi_0(s)^{-1},
\end{align*}
as $\zeta \to 0$.
The above formula, together with \eqref{psi-B} and \eqref{a-1-b-1}, gives us
\begin{equation}
  \frac{d\Psi_0(s)}{ds} \cdot \Psi_0(s)^{-1} =  B(0;s) = \begin{pmatrix}
    0 &1\\ 2a_1'(s) & 0
  \end{pmatrix},
\end{equation}
where $a_1'(s) = \frac{da_1(s)}{ds}$. Thus, from \eqref{def:H-1}, it follows that
\begin{multline}\label{def:H-1-diff}
  \frac{d }{d\lambda}H_1(\lambda) \cdot H_1(\lambda)^{-1}
  \\ =  n^{2/3} \frac{d f(\lambda)}{d \lambda}\left[ \begin{pmatrix}
    a_1(s)-\frac{s^2}{4} & 1 \\ a'_1(s) +\frac{s}{2}-(a_1(s)-\frac{s^2}{4})^2 & -a_1(s)+\frac{s^2}{4}
  \end{pmatrix}\right]_{s=n^{2/3}f(\lambda)}.
\end{multline}
For the last term in \eqref {def:H-diff-1}, we recall from the estimate \eqref{def:p-lambda}
to obtain
\begin{equation}\label{p-H-est}
\frac{d}{d\lambda}p(\lambda;n) \cdot H_1(\lambda)\sigma_3H_1(\lambda)^{-1}= O(n^{-2/3})H_1(\lambda)\sigma_3H_1(\lambda)^{-1}.
\end{equation}
We now show that, for large $n$,
\begin{equation}\label{eq:estH1}
H_1(\lambda)\sigma_3H_1(\lambda)^{-1}=O(n^{1/6})
\end{equation}
Indeed,  we note from \eqref{Psi-origin}, \eqref{D-expand} and  \eqref{eq:Psi-large s-0} that
\begin{align}\label{Q-est}
\Psi_0(s)=& \left(
              \begin{array}{cc}
                1 & 0 \\
                a_1(s) -\frac{s^2}{4}-\frac{c_1}{\sqrt{s}} & 1 \\
              \end{array}
            \right)
            e^{-\frac{1}{4}\pi i\sigma_3}s^{-\frac{1}{4}\sigma_3}(I+O(1/s)) \frac{I+i\sigma_1}{\sqrt{2}}e^{-\frac{2}{3}s^{3/2}\sigma_3}e^{-\epsilon_1(s)\sigma_3},
\end{align}
where the error bound $O(1/s)$ is uniform for $s>c_0$ for certain big enough constant $c_0$, and
\begin{equation}\label{def:epsilon-1}
\epsilon_1(s)=\lim_{\zeta\to0}\left(\sqrt{1+\frac{\zeta}{s}}\sum_{k=1}^{2m}\frac{c_k}{\zeta^k}-\sum_{k=1}^{2m}\frac{\tau_k}{\zeta^k}
\right)
\end{equation}
is bounded due to \eqref{q-zeta-approx}. Thus, for $n^{2/3}f(\lambda)>c_0$ and $1-cn^{-2/3}<\lambda<1+dn^{-1/3}$, we obtain by inserting \eqref{Q-est} into \eqref{def:H-1}
\begin{equation}\label{H-1-est-0}
H_1(\lambda)\sigma_3H_1(\lambda)^{-1}=O(n^{1/3}\sqrt{f(\lambda)})=O(n^{1/6}), \qquad \textrm{as } n \to +\infty,
\end{equation}
where we have also made use of the large $s$ behavior of $a_1(s)$ given in \eqref{eq:a 1 asy}.
If $n^{2/3}f(\lambda)<c_0$ and $1-cn^{-2/3}<\lambda<1-dn^{-1/3}$,  then $n^{2/3}f(\lambda)$ is uniformly bounded in $n$.
Since both $a_1(s)$ and $\Psi_0(s)$ are smooth in $s$, it is immediate from \eqref{def:H-1} that
\begin{equation}\label{H-1-est-1}
H_1(\lambda)\sigma_3H_1(\lambda)^{-1}=O(1).
\end{equation}
The estimate \eqref{eq:estH1} then follows from  \eqref{H-1-est-0} and \eqref{H-1-est-1}.

Hence, combining \eqref{def:H-0-diff}, \eqref{p-H-est} and \eqref{eq:estH1}, it is readily seen that
\begin{multline}\label{eq:H-det-est}
\det \left(\frac{dH(\lambda)}{d\lambda}\cdot H(\lambda)^{-1}\right)
\\
=\det\left( \frac{dH_1(\lambda)}{d\lambda} \cdot H_1(\lambda)^{-1} +
\begin{pmatrix}
  O(1) & O(n^{-2/3}) \\ O(1) & O(1)
\end{pmatrix}
+O(n^{-1/2})\right).
\end{multline}
This, together with \eqref{def:H-1-diff}, implies that
\begin{multline}
  \det\left(\frac{dH(\lambda)}{d\lambda} \cdot H(\lambda)^{-1}\right)
  %& = - \left(n^{2/3} \frac{d f(\lambda)}{d \lambda}\right)^2 \biggl( a'_1(n^{2/3}f(\lambda)) +\frac{n^{2/3}f(\lambda)}{2} + O(n^{-2/3}) + O(s^{3/2} n^{-4/3}) \biggr) \nonumber \\
   =  - \left(n^{2/3} \frac{d f(\lambda)}{d \lambda}\right)^2 \left( a'_1(n^{2/3}f(\lambda)) +\frac{n^{2/3}f(\lambda)}{2} + O(n^{-2/3}) \right). \label{eq:H-det}
\end{multline}
%where $a_1'(s) = \frac{da_1(s)}{ds}$ and the error bound is uniform for $\lambda=1+s/n^{2/3}$, $c<s<dn^{1/3}$ with constants $c$ and $d$.
%= \left( -4n^{4/3} - \frac{4s}{5} n^{2/3}  + \frac{12s^2}{175} \right) \left(a_1'(s)+\frac{s}{2}\right) +O(n^{2/3}) \label{eq:H-det}
Finally, we note from the definitions of $f(z)$  and $\phi(z)$ in \eqref{def:f}  and \eqref{phi} that
\begin{equation}
  4n^2(\lambda^2-1) = n^2\left(\frac{df(\lambda)}{d\lambda}\right)^2f(\lambda).
  \end{equation}
Substituting the above two formulas into \eqref{diff identity-2} yields
\begin{equation}\label{eq:Hakel-deff-2}
    \frac{d^2}{d\lambda^2}\ln Z_n(\lambda)=-\left(n^{2/3} \frac{d f(\lambda)}{d \lambda}\right)^2 \biggl( a'_1(n^{2/3}f(\lambda)) -\frac{n^{2/3}f(\lambda)}{2} + O(n^{-2/3}) \biggr),
\end{equation}
where the error bound is uniform for $1+cn^{-2/3}<\lambda<1+dn^{-1/3}$ with any choice of  constants $c<0$ and $d>0$.

This completes the proof of Lemma \ref{thm:Log-Partition function-asy}.
\end{proof}

%%%%%%%%%%%%%%%%%%%%%%%%%%%%%%%%%%%%%%%%%%%%%%%%%%%%%%%%%%%%%%%%%%%%%%%%%%%%%

\section{Asymptotic analysis of the RH problem for $Y$ with $\lambda\geq1+n^{-2/5}$}
\label{sec:AARHY2}
In this section, we study the large $n$ behavior of the RH problem for $Y$ with $\lambda\geq1+n^{-2/5}$ and $n^{\frac{2k}{3}+1} t_k$, $k=1,2,\ldots,2m$, bounded. At the end, the asymptotics of the partition function is presented in this case. This result, together with Lemma \ref{thm:Log-Partition function-asy}, will finally lead us to the proof of Theorem \ref{thm:partition asymptotics}.

\subsection{The transformations $Y \to  \widetilde{T} \to  \widetilde{S}$}
In the present case, the first transformation is defined by
\begin{equation}\label{Y to T-2}
 \widetilde{T}(z)= e^{-\frac 1 2 nl \sigma_3} Y(z) e^ {n \left (\frac 1 2 l
-\hat{g}(z)\right )\sigma_3} \quad z\in
\mathbb{C} \setminus \mathbb{R},
\end{equation}
where the $\hat{g}$-function is defined in \eqref{g}. One may compare the above function $\widetilde{T}(z)$ with $T(z)$ defined in \eqref{Y to T}. It is easily seen that $ \widetilde{T}$ satisfies the following RH problem.

\subsubsection*{RH problem for $\widetilde{T}$}
\begin{enumerate}
  \item[\rm(a)]  $\widetilde{T}(z)$ is analytic in
  $\mathbb{C} \setminus \mathbb{R}$.

  \item[\rm(b)]  $\widetilde{T}(z)$  satisfies the jump condition
 \begin{equation}\label{Jump-T-1}\widetilde{T}_+(x)=\widetilde{T}_-(x)
 \left\{
\begin{array}{ll}
 \begin{pmatrix}
                                 1 & e^{-2n \phi(x)}e^{-2\widetilde{V}(x)} \\
                                 0 & 1
 \end{pmatrix},&   x\in(1,+\infty), \\
  [.4cm]
  \begin{pmatrix}
                                 e^{2n\phi_+(x)} & e^{-2\widetilde{V}(x)} \\
                                 0 & e^{2n \phi_-(x)}
  \end{pmatrix},&   x\in(-1, 1), \\
 [.4cm]
  \begin{pmatrix}
                                 1 &  e^{-2n \phi(x)}e^{-2\widetilde{V}(x)} \\
                                 0 & 1
  \end{pmatrix}, &  x\in(-\infty,-1),
\end{array}
\right .
\end{equation}
where
\begin{equation}\label{def: V-tile}
\widetilde{V}(x)=\frac{n}{2}\sum_{k=1}^{2m}t_k/(x-\lambda)^k.
\end{equation}

\item[\rm(c)]  As $z\to \infty$, we have
  \begin{equation}\label{T-infinity-1}
  \widetilde{T}(z)=I+O(1/z).
  \end{equation}

\item[\rm(d)] $\widetilde{T}(z)$ is bounded at $z = \lambda$.
\end{enumerate}

We next open lens around $[-1,1]$ as shown in Figure \ref{contour-for-tildeS} and introduce the following transformation:
\begin{equation}\label{T-S-1}
\widetilde{S}(z)=\left \{
\begin{array}{ll}
  \widetilde{T}(z), & \mbox{for $z$ outside the lens,}
  \\
  [.4cm]
  \widetilde{T}(z) \begin{pmatrix}
                                 1 & 0 \\
                                  - e^{2n\phi(z)} e^{2\widetilde{V}(z)} & 1
                                  \end{pmatrix} , & \mbox{for $z$ in the upper lens,}\\
                                  [.4cm]
\widetilde{T}(z) \begin{pmatrix}
                                 1 & 0 \\
                                  e^{2n\phi(z)} e^{2\widetilde{V}(z)} & 1
                                 \end{pmatrix} , & \mbox{for $z$ in the lower lens. }
\end{array}\right.
\end{equation}
Then, $\widetilde{S}$ satisfies the following RH problem.

\subsubsection*{RH problem for $\widetilde{S}$ }
\begin{enumerate}
  \item[\rm(a)]  $\widetilde{S}(z)$ is analytic in
  $\mathbb{C}\setminus \Sigma_{\widetilde{S}}$, where the contour $ \Sigma_{\widetilde{S}}$ is illustrated in Figure \ref{contour-for-tildeS}.

  \item[\rm(b)]  $\widetilde{S}(z)$  satisfies the jump condition
 \begin{equation}
 \widetilde{S}_+(z)=\widetilde{S}_-(z)\left\{
\begin{array}{ll}
 \begin{pmatrix}
                                 1 & e^{-2n \phi(z)}e^{-2\widetilde{V}(z)} \\
                                 0 & 1
 \end{pmatrix},&   z \in(1,+\infty)\cup (-\infty,-1),
 \\
 [.4cm]
  \begin{pmatrix}
                                 1 & 0 \\
                                 e^{2\widetilde{V}(z)} e^{2n \phi(z)}& 1
  \end{pmatrix},&   \mbox {$z\in \Sigma_{\widetilde S}^+ \cup \Sigma_{\widetilde S}^- $, }
  \\
  [.4cm]
   \begin{pmatrix}
                                 0 & e^{-2\widetilde{V}(z)} \\
                                 - e^{2\widetilde{V}(z)}& 0
   \end{pmatrix}, &  z \in(-1, 1).
\end{array}
\right .
\end{equation}

\item[\rm(c)] As $z\to \infty$, we have
  \begin{equation}
  \widetilde{S}(z)=I+O(1/z).
  \end{equation}

\item[\rm(d)] $\widetilde{S}(z)$ is bounded at $z = \lambda$.
\end{enumerate}

 \begin{figure}[h]
 \begin{center}
   \includegraphics[width=8cm]{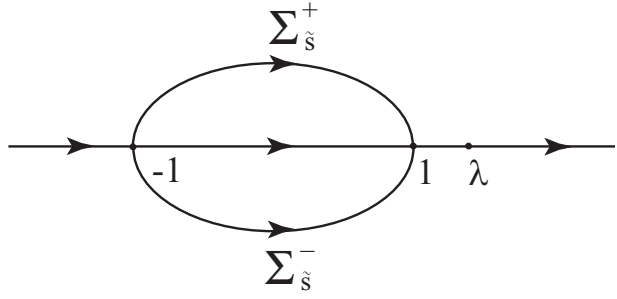}
 \end{center}
  \caption{Contour $\Sigma_{\widetilde{S}}$ for the RH problem for $\widetilde{S}$.}
 \label{contour-for-tildeS}
\end{figure}

%--------------------------------------------------------------------------------------------------------------
\subsection{Outer and local parametrices}
Outside a small disk centered at $\lambda$ with radius $r>0$, the solution to the RH problem for $\widetilde{S}(z)$ can be approximated by
$S_0(z) e^{\widetilde{V}(z)\sigma_3}$, where $S_0(z)$ is the solution to the above RH problem for $S$ with all the parameters $t_k$ in \eqref{def: V-tile} vanish.
However, since $S_0(z) e^{\widetilde{V}(z)\sigma_3}$ possesses essential singularity at $z = \lambda$, it violates item (d) in the above RH problem. We therefore simply use $S_0(z)$ as the local parametrix.

The RH problem for $S_0$ is actually the one in the asymptotic analysis of the classical Hermite polynomials via the RH approach. Indeed, we have
\begin{equation}\label{def: S-0}
S_0(z)=\left\{ \begin{array}{ll}
             R_0(z)E_0(z)\Phi_{\Ai}(n^{2/3}f(z))e^{n\phi(z)\sigma_3}, & \Re z >0 \quad \mbox{and} \quad z\in   U(1,\delta),  \\
              R_0(z)N_0(z), & \Re z >0 \quad \mbox{and }\quad  z\not\in   U(1,\delta),
                           \end{array}\right .
\end{equation}
where $\delta>0$ is a small positive constant, $E_0(z)$ is analytic in $U(1,\delta)$,  $\Phi_{\Ai}$ is the Airy parametrix given by \eqref{Airy-model-solution},  $\phi(z)$ and $f(z)$ are defined in \eqref{phi} and \eqref{def:f}. In \eqref{def: S-0}, the  function $N_0(z)$ is defined by
\begin{equation}\label{N-0}
N_0(z)=\frac{I-i\sigma_1}{\sqrt{2}}\beta(z)^{-\sigma_3}\frac{I+i\sigma_1}{\sqrt{2}}=\left(
  \begin{array}{cc}
    \frac {\beta(z) + \beta^{-1}(z)} {2}&\frac {\beta(z) - \beta^{-1}(z)} {2i} \\
    -\frac {\beta(z) - \beta^{-1}(z)} {2i} &\frac {\beta(z) + \beta^{-1}(z)} {2} \\
  \end{array}
\right),
\end{equation}
where $ \beta(z)=\left ( \frac {z-1}{z+1} \right )^{1/ 4}$ is analytic in $\mathbb{C}\setminus [-1, 1]$ and $\beta(z)\sim 1$ as $z\to\infty$. The function $R_0(z)$ is analytic in $U(1,\delta)$ and
\begin{equation}\label{eq:R-0 est}
R_0(z)=I+O(1/n)\end{equation}
for $n$ large.

For later use, we also note that
\begin{align}
  S_0(z) \sim R_0(z)  N_0(z) \left(I+\sum_{k=1}^{\infty}\frac {A_k}{(n\phi(z))^{k}}\right), \label{eq:S-0-expan}
\end{align}
where $A_k $ are some constant matrices; see \cite{dkmv2}. The above expansion holds in a possibly shrinking neighborhood of $1$ as long as the condition $\lim_{n\to +\infty}n\phi(z) = \infty$ is satisfied.

%%%%%%%%%%%%%%%%%%%%%%%%%%%%%%%%%%%%%%%%%%%%%%%%%%%%%%%%%%%%%%%%%%%
\subsection{Final transformation}
With outer and local parametrices given by $S_0(z) e^{\widetilde{V}(z)\sigma_3}$ and $S_0(z)$, respectively,
we define the final transformation as follows:
\begin{equation}\label{def: R}
\widetilde{R}(z)=\left\{ \begin{array}{ll}
               \widetilde{S}(z)e^{-\widetilde{V}(z)\sigma_3} S_0(z)^{-1}, & z\not\in   U(\lambda,r),  \\
               \widetilde{S}(z) S_0(z)^{-1}, & z\in    U(\lambda,r).
             \end{array}\right .
\end{equation}

As usual, we intend to show that the function $\widetilde R$ tends to $I$ as $n\to \infty$ by considering the large $n$ behavior of its jump. We start with the case that $\lambda\in (1+n^{-2/5}, 1+\epsilon_0)$, where $\epsilon_0>0$ is a fixed and small constant. In this case, we choose  the constant  $\delta> 1+ \epsilon_0$ in \eqref{def: S-0} and $r=\frac{\lambda-1}{2}$ so that the neighborhood $U(\lambda,r)$ is contained in $U(1,\delta)$. Note that, as $\lambda$ may tend to the endpoint $z=1$ with a rate $n^{-2/5}$, $U(\lambda,r)$ could be a shrinking neighborhood, which requires some careful estimates.

Since $U(\lambda,r)$ does not intersect the upper and lower lens in the contour $\Sigma_{\widetilde{S}}$, we have that $\widetilde{R}$ satisfies the following RH problem.
\subsubsection*{RH problem for $\widetilde{R}$}
\begin{enumerate}
  \item[\rm(a)] $\widetilde{R}(z)$ is analytic in $\mathbb{C}\backslash \{ (\lambda-r,\lambda+r) \cup \partial U(\lambda,r) \}$; see Figure \ref{contour-for-tilde R} for an illustration of the contour.
  \item[\rm(b)] $\widetilde{R}(z)$  satisfies the jump condition
 \begin{equation} \label{eq: J-tildeR}
 \widetilde{R}_+(z) =\widetilde{R}_-(z)J_{\widetilde {R}}(z),
 \end{equation}
 where
 \begin{equation}\label{def:JtildeR}
  J_{\widetilde {R}}(z) = \left\{
\begin{array}{ll}
 S_{0,-}(z) \begin{pmatrix}
   1 & e^{-2n \phi(z)} (e^{-\widetilde{V}(z)} - 1) \\ 0 & 1
 \end{pmatrix} S_{0,-}(z)^{-1}, &   z\in(\lambda-r,\lambda+r), \\
  S_0(z)e^{-\widetilde{V}(z)\sigma_3}S_0(z)^{-1} , &   \mbox {$z \in \partial U(\lambda,r)$}.
\end{array}
\right .
\end{equation}

\item[\rm(c)] As $z \to \infty$,
 \begin{equation}
  \widetilde{R}(z)=I+O(1/z).
  \end{equation}
\end{enumerate}

\begin{figure}[h]
 \begin{center}
   \includegraphics[width=3cm]{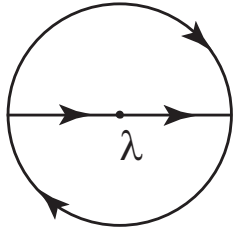}
 \end{center}
  \caption{Contour for the RH problem of $\widetilde{R}$.}
 \label{contour-for-tilde R}
\end{figure}

Because $e^{-n \phi(z)}$ is exponentially small as $n \to \infty$  and $e^{-\widetilde{V}(z)} - 1$ is bounded for $z \in (\lambda-r,\lambda+r)$ , it is easily seen that, if $z\in(\lambda-r,\lambda+r)$,
\begin{equation*}
  J_{\widetilde{R}}(z) = I + S_{0,-}(z) \begin{pmatrix}
   0 & e^{-2n \phi(z)} (e^{-\widetilde{V}(z)} - 1) \\ 0 & 0
 \end{pmatrix} S_{0,-}(z)^{-1}
\end{equation*}
tends to the identity matrix exponentially fast as $n\to \infty$.
For  $z\in \partial U(\lambda,r)$, we have $|n\phi(z)|\to \infty$, and by \eqref{eq:S-0-expan} and \eqref{def:JtildeR}, it follows that
\begin{align}
  J_{\widetilde{R}}(z) & =I+S_0(z)(e^{-\widetilde{V}(z)\sigma_3}-I)S_0(z)^{-1} \nonumber \\
  & = I+R_0(z)N_0(z)(e^{-\widetilde{V}(z)\sigma_3}-I)N_0(z)^{-1}R_0(z)^{-1}+O\left(\frac{1}{n^{5/3}(\lambda-1)^3} \right),
  \label{eq: J-R -exp}
\end{align}
where use have also been  made of the face that  $n\phi(z)=O(n(\lambda-1)^{3/2}), N_0(z)=O(1/(\lambda-1)^{1/4})$ as $z\to \lambda$.
Using \eqref{N-0}, we further have
\begin{multline}\label{eq: J-0-exp}
N_0(z)(e^{-\widetilde{V}(z)\sigma_3}-I)N_0(z)^{-1}
\\=(\cosh(\widetilde{V}(z))-1)I+\sinh (\widetilde{V}(z))\frac{I-i\sigma_1}{\sqrt{2}}\left(
                                                                  \begin{array}{cc}
                                                                    0 & i\sqrt{\frac{z+1}{z-1}} \\
                                                                   - i\sqrt{\frac{z-1}{z+1}} & 0 \\
                                                                  \end{array}
                                                                \right)\frac{I+i\sigma_1}{\sqrt{2}}.
\end{multline}
On account of  \eqref{def: V-tile}, we have, for $z \in \partial U(\lambda,r)$
\begin{align}
\cosh(\widetilde{V}(z))-1 & =O\left(\frac{1}{n^{4/3}(\lambda-1)^2}\right), \qquad \textrm{as } n \to \infty,
\\
\sinh(\widetilde{V}(z))&=\frac{nt_1}{2}\frac{1}{z-\lambda} +O\left(\frac{1}{n^{4/3}(\lambda-1)^2}\right), \qquad \textrm{as } n \to \infty, \label{eq:sinhest}
\end{align}
where we also use the condition that all $n^{\frac{2k}{3}+1} t_k$, $k=1,2,\ldots,2m$, are bounded.
Inserting \eqref{eq: J-0-exp}--\eqref{eq:sinhest} into \eqref{eq: J-R -exp} yields
\begin{equation}\label{eq: J-0-est}
N_0(z)(e^{-\widetilde{V}(z)\sigma_3}-I)N_0(z)^{-1}=\widetilde{J}_1(z)+O\left(\frac{1}{n^{4/3}(\lambda-1)^{5/2}}\right),
\end{equation}
 where
\begin{equation}\label{eq: J-1}
\widetilde{J}_1(z)=\frac{nt_1}{2}\frac{1}{z-\lambda}\frac{I-i\sigma_1}{\sqrt{2}} \left(
                                                                  \begin{array}{cc}
                                                                    0 & i\sqrt{\frac{z+1}{z-1}} \\
                                                                   - i\sqrt{\frac{z-1}{z+1}} & 0 \\
                                                                  \end{array}
                                                                \right)\frac{I+i\sigma_1}{\sqrt{2}}.
\end{equation}
Since  $\lambda\in (1+n^{-2/5}, 1+\epsilon_0)$ and $n^\frac{5}{3}t_1$ is bounded, it is readily seen that
\begin{equation}
\widetilde{J}_1(z)=
O \left(\frac{nt_1}{(\lambda-1)^{3/2}} \right)=O \left(\frac{1}{n^{2/3}(\lambda-1)^{3/2}} \right)=O(n^{-1/15}),
\end{equation}
uniformly for $z\in\partial U(\lambda,r)$. Substituting the estimates \eqref{eq:R-0 est} and \eqref{eq: J-0-est} into \eqref{eq: J-R -exp}, we obtain
\begin{equation}\label{eq: J-R-expand}
J_{\widetilde{R}}(z) = I + \widetilde{J}_1(z)+O\left(\frac{1}{n^{4/3}(\lambda-1)^{5/2}}\right),
\end{equation}
where $\widetilde{J}_1(z)$ is defined in \eqref{eq: J-1}. Note that,  the radius $r=\frac{\lambda-1}{2}$ may tend to $0$ as $n \to \infty$, it is then
more convenient to introduce  the centering and scaling of variable $$z\to \lambda +rz ,$$ and rewrite \eqref{eq: J-R-expand} as
\begin{equation}\label{eq: J-R-expand-cent}
J_{\widetilde{R}}( \lambda +rz) = I+J_1( \lambda +rz)+O\left(\frac{1}{n^{4/3}(\lambda-1)^{5/2}}\right),
\end{equation}
where the error bounded is uniform for $|z|=1$.
By a standard analysis of the small norm RH problems \cite{dkmv2},  it follows from \eqref{eq: J-R-expand-cent} that
\begin{equation}\label{eq: R-expand}
\widetilde{R}( \lambda +rz) = I+\widetilde{R}^{(1)}( \lambda +rz)+O\left(\frac{1}{n^{4/3}(\lambda-1)^{5/2}}\right),
\end{equation}
 where the error bounded is uniform for $z$ bounded away from $|z|=1$ and
\begin{align}\label{def:R-hat}
&\widetilde{R}^{(1)}(z)
\nonumber \\
&=\frac{1}{2\pi i} \oint_{\partial U(\lambda,r)}
  \frac{\widetilde {J}_1(w)}{w-z} dw
  \nonumber \\
  &=
\begin{cases}
   \frac{nt_1}{2(z-\lambda)}\frac{I-i\sigma_1}{\sqrt{2}}\left(
                                                                  \begin{array}{cc}
                                                                    0 & i\sqrt{\frac{\lambda+1}{\lambda-1}} \\
                                                                   - i\sqrt{\frac{\lambda-1}{\lambda+1}} & 0 \\
                                                                  \end{array}
                                                                \right)
\frac{I+i\sigma_1}{\sqrt{2}}, &  z\not\in U(\lambda,r), \\
\quad \\
 \frac{nt_1}{2(z-\lambda)}\frac{I-i\sigma_1}{\sqrt{2}}\left(
                                                                  \begin{array}{cc}
                                                                    0 & -i\left(\sqrt{\frac{z+1}{z-1}}-\sqrt{\frac{\lambda+1}{\lambda-1}}\right) \\
                                                                   i\left(\sqrt{\frac{z-1}{z+1}}-\sqrt{\frac{\lambda-1}{\lambda+1}} \right)& 0 \\
                                                                  \end{array}
                                                                \right)
\frac{I+i\sigma_1}{\sqrt{2}}, & z\in U(\lambda,r).
\end{cases}
\end{align}
Moreover, if $z\to \infty$, it follows from \eqref{def:R-hat} that
\begin{equation}\label{eq: R-(1)-expand}
\widetilde{R}^{(1)}(z)= \frac{\widetilde{R}^{(1)}_1}{z} +O(1/z^2),
\end{equation}
where
\begin{equation}\label{def:R11}
\widetilde{R}^{(1)}_1=\frac{nt_1}{2}\frac{I-i\sigma_1}{\sqrt{2}}\left(
                                                                  \begin{array}{cc}
                                                                    0 & i\sqrt{\frac{\lambda+1}{\lambda-1}} \\
                                                                   - i\sqrt{\frac{\lambda-1}{\lambda+1}} & 0 \\
                                                                  \end{array}
                                                                \right)
\frac{I+i\sigma_1}{\sqrt{2}}.\end{equation}
By \eqref{eq: R-expand}, we see that as $z\to \infty$,
\begin{equation}\label{eq: R-expand-infinity}
\widetilde{R}(z) = I+\frac{\widetilde{R}_1}{z}+O(1/z^2).
\end{equation}
Comparing \eqref{eq: R-(1)-expand}, \eqref{eq: R-expand-infinity}, \eqref{def:R11}  with \eqref{eq: R-expand} gives us
\begin{equation}\label{eq: R-1}
\widetilde{R}_1 =\widetilde{R}^{(1)}_1+O\left(\frac{1}{n^{4/3}(\lambda-1)^{3/2}}\right).
\end{equation}
This, together with \eqref{eq: R-(1)-expand}, implies that
\begin{equation}\label{eq:R-1}
(\widetilde{R}_1)_{11}=-\frac{nt_1}{2}\frac{\lambda}{\sqrt{\lambda^2-1}}+O\left(\frac{1}{n^{4/3}(\lambda-1)^{3/2}}\right),  \qquad n\to \infty,
\end{equation}
for $\lambda\in (1+n^{-2/5}, 1+\epsilon_0)$.

The case when $\lambda \geq 1+\epsilon_0$ can be treated in a similar manner.  In this case, the pole is bounded away from the right endpoint of the limiting spectrum. Then, we choose the radius $\delta<\lambda/2$  and $r=(\lambda - 1)/2$ in \eqref{def: S-0} and \eqref{def: R}, respectively, such that $U(\lambda,r) \cap U(1,\delta) = \phi$. The approximation solution is simply given by  $S_0(z)=R_0(z)N_0(z)$ for $z\in U(\lambda,r)$.  After some  direct computations, we conclude that
\begin{equation}\label{eq:R-1-case-2}
(\widetilde{R}_1)_{11}=-\frac{nt_1}{2}\frac{\lambda}{\sqrt{\lambda^2-1}}+O\left(\frac{1}{n^{4/3}\lambda^2}\right), \qquad n\to \infty.
\end{equation}

%-----------------------------------------------------------------------------------------------------------------------------------------------------------------------------------------------------------------
\subsection{Asymptotics of the partition function}
As a consequence of the asymptotic analysis just performed, we obtain the following asymptotics of the partition function with the aid of
differential identity \eqref{diff identity}.

\begin{lem}\label{thm:Partition function-Asy-outer}
If $\lambda\in[ 1+n^{-2/5},+\infty)$ and
$n^{\frac{2}{3}k+1}t_k$, $k=1,2,\ldots,2m$ are bounded,  we have
\begin{equation}\label{Z-n-asy-out}
Z_n(\lambda)= Z_n^{GUE}\exp\left( \frac{2n^2t_1}{\lambda+\sqrt{\lambda^2-1}}\right)\left(1+O\left(\frac{1}{n^{1/3}\sqrt{\lambda-1}}\right)\right),\qquad n\to\infty,
\end{equation}
where $ Z_n^{GUE}$ is the partition function of GUE given in \eqref{zn-gue}.
\end{lem}

\begin{proof}
From the differential identity \eqref{diff identity}, we need to compute the residue term $Y_1$ in \eqref{Y-infinity}. By tracing back the sequence of transformations in \eqref{Y to T-2}, \eqref{T-S-1} and \eqref{def: R}, it follows that
\begin{equation} \label{Y tracing back-1}
Y(z)=e^{\frac 12 nl\sigma_3}\widetilde{R}(z)S_0(z)e^{\widetilde{V}(z)\sigma_3}e^{n \hat{g}(z)\sigma_3-\frac 12 nl \sigma_3}, \qquad z\not \in U(\lambda,r).
\end{equation}
As $z \to \infty$, we have
\begin{equation}\label{eq: expo-exp}
e^{\widetilde{V}(z)\sigma_3}=I+\frac{nt_1}{2}\frac{1}{z} \sigma_3 +O(1/z^2).
\end{equation}
For the large $z$ behavior of $S_0(z)e^{ \hat{g}(z) \sigma_3}$, we note that, when $t_k=0$, the polynomials $\pi_n(x)$ in \eqref{eq:Y} reduce to the classical Hermite polynomials, of which the sub-leading coefficient vanishes. This gives us
\begin{equation} \label{eq: Y-1-0}
  S_0(z)e^{n \hat{g}(z) \sigma_3} = I + \frac{\mathcal{S}_1}{z} + O(1/z^2)
\end{equation}
with $(\mathcal{S}_1)_{11} = 0$.
It then follows from  \eqref{eq:R-1}--\eqref{eq: Y-1-0} that
\begin{align}
  (Y_1)_{11}& =  (\widetilde{R}_1)_{11} + (\mathcal{S}_1)_{11} +\frac{nt_1}{2} \nonumber \\
  &=\frac{nt_1}{2} \left( 1-\frac{\lambda}{\sqrt{\lambda^2-1}} \right)+O\left(\frac{1}{n^{4/3}(\lambda-1)^{3/2}}\right), \label{Y1}
\end{align}
where the error bound is uniform for $\lambda\in[1+ n^{-2/5}, \infty)$.
By the differential identity \eqref{diff identity}, we obtain
\begin{equation}\label{log d partition-asy }
\frac{d}{d\lambda} \ln Z_n(\lambda)= 2n^2t_1 \left(1-\frac{\lambda}{\sqrt{\lambda^2-1}} \right)+O\left(\frac{1}{n^{1/3}(\lambda-1)^{3/2}}\right).
\end{equation}
Integrating the above equation from $\lambda$ to $+\infty$ gives us
\begin{equation}\label{ partition-asy }
 \ln Z_n(\lambda)- \ln Z_n^{GUE}=\frac{2n^2t_1}{\lambda+\sqrt{\lambda^2-1}}+O\left(\frac{1}{n^{1/3}\sqrt{\lambda-1}}\right),
 \end{equation}
where $Z_n^{GUE}$ is the partition function for GUE as defined in \eqref{zn-gue} and we have also made use of the fact that $Z_n(\lambda) \to Z_n^{GUE}$ as $\lambda\to\infty$.

This completes the proof of Lemma \ref{thm:Partition function-Asy-outer}.
\end{proof}

We are now ready to prove Theorem \ref{thm:partition asymptotics}.
%---------------------------------------------------------------------------------------------------------------------------------------------------------------------------------------------------------------
\section{Proof of  Theorem \ref{thm:partition asymptotics}} \label{sec:pfasyparti}
We integrate both sides of \eqref{log-Z-n-asy} with respect to $\lambda$ and obtain
\begin{align}%\label{eq:Hakel-deff-0}
    & \frac{d}{d\lambda}\ln Z_n(\lambda) - \frac{d}{d\lambda}\ln Z_n(\lambda)\biggl|_{\lambda = \lambda_0}  \nonumber \\
    & \qquad  = -\int_{\lambda_0}^\lambda \left(n^{2/3} f'(\xi) \right)^2 \biggl( a'_1(n^{2/3}f(\xi)) -\frac{n^{2/3}f(\xi)}{2} \biggr) d \xi + O(n^{1/3}) ,
\end{align}
where both $\lambda$ and  $\lambda_0$ belong to $ (1+cn^{-2/3}, 1+dn^{-1/3} )$ such that the above $O(n^{1/3})$ term holds uniformly. Applying integration by parts in the above integral gives us
\begin{align}
&-\int_{\lambda_0}^\lambda \left(n^{2/3} f'(\xi) \right)^2 \biggl( a'_1(n^{2/3}f(\xi)) -\frac{n^{2/3}f(\xi)}{2} \biggr) d \xi\nonumber\\
 &\qquad = -n^{2/3}f'(\xi) \left(a_1(n^{2/3}f(\xi))-\frac{n^{4/3}f(\xi)^2}{4} \right) \biggr|_{\xi=\lambda_0}^\lambda \nonumber\\
 & \qquad \qquad + n^{2/3} \int_{\lambda_0}^\lambda  f''(\xi) \left(a_1(n^{2/3}f(\xi))-\frac{n^{4/3}f(\xi)^2}{4} \right) d\xi.\nonumber
\end{align}
From the definition of $f(z)$ in \eqref{def:f} and the asymptotics of $a_1(s)$ in \eqref{eq:a 1 asy}, one can see that the integrand in the above integral is uniformly bounded for $\xi \in [\lambda, \lambda_0]$. As $|\lambda - \lambda_0| = O(n^{-1/3})$, the second term in the above formula is of order $O(n^{1/3})$ uniformly for $\lambda, \lambda_0 \in  (1+cn^{-2/3}, 1+dn^{-1/3} )$. Then, the above two formulas give us
\begin{equation}
  \frac{d}{d\lambda}\ln Z_n(\lambda) = -n^{2/3}f'(\lambda) \left(a_1(n^{2/3}f(\lambda))-\frac{n^{4/3}f(\lambda)^2}{4} \right)+ d_0 +O(n^{1/3}),
\end{equation}
where $d_0$ is the following $\lambda$-independent constant
\begin{equation}
  d_0 =  \frac{d}{d\lambda}\ln Z_n(\lambda)\biggl|_{\lambda = \lambda_0} + n^{2/3}f'(\lambda_0) \left(a_1(n^{2/3}f(\lambda_0))-\frac{n^{4/3}f(\lambda_0)^2}{4} \right).
\end{equation}
We may choose $\lambda_0=1+n^{-2/5}$, then \eqref{log d partition-asy } gives us
\begin{align*}
  \frac{d}{d\lambda}\ln Z_n(\lambda)\biggl|_{\lambda = \lambda_0} & = 2n^2 t_1 \left(1 - \frac{n^{1/5} + n^{-1/5}}{\sqrt{2+n^{-2/5}}} \right) +O(n^{4/15}) \\
  & = -\sqrt{2} \, \tau_1 n^{8/15} + O(n^{1/3}),
\end{align*}
where use has been made of the condition $t_1 = \tau_1  n^{-5/3} + O(n^{-7/3})$; see \eqref{eq:tkscaling}. Recalling \eqref{def:f}, we get $f(\lambda_0) = 2 n^{-2/5} +O(n^{-4/5})$ and $f'(\lambda_0) = 2 +O(n^{-2/5})$. Then, we have from the asymptotics of $a_1(s)$ in \eqref{eq:a 1 asy}
\begin{align*}
   n^{2/3}f'(\lambda_0) \left(a_1(n^{2/3}f(\lambda_0))-\frac{n^{4/3}f(\lambda_0)^2}{4} \right) %&= \sqrt{2} \, n^{8/15} \biggl(1+O(n^{-2/5}) \biggr) \biggl(\tau_1 + O(n^{-4/15}) \biggr) \\
   = \sqrt{2} \, \tau_1 n^{8/15} + O(n^{4/15}).
\end{align*}
The above three formulas imply
\begin{equation}
  d_0=O(n^{1/3}).
\end{equation}
Hence, we have
\begin{equation}\label{eq:Hakel-deff-1}
     \frac{d}{d\lambda}\ln Z_n(\lambda)=- n^{2/3}f'(\lambda) \left(a_1(n^{2/3}f(\lambda))-\frac{n^{4/3}f(\lambda)^2}{4} \right)+O(n^{1/3}),
\end{equation}
uniformly for $1+cn^{-2/3}<\lambda<1+dn^{-1/3}$. A further integration of \eqref{eq:Hakel-deff-1} on both sides with respect to $\lambda$ then gives us
\begin{equation}\label{Z-n-asy}
Z_n(\lambda) = Z_n(\lambda_0)\exp\left[\int_{n^{2/3}f(\lambda)}^{n^{2/3}f(\lambda_0)} \left(a_1(t)-\frac{t^2}{4} \right) dt\right](1+o(1)),
\end{equation}
where the error bound is uniform for $1+cn^{-2/3}<\lambda,\lambda_0<1+d_n$ with $0<d_n = o(n^{-1/3})$.

We may choose $\lambda_0=1+n^{-2/5}$ again and estimate the two terms on the right hand side of \eqref{Z-n-asy}. For $Z_n(\lambda_0)$, it follows from \eqref{Z-n-asy-out} that
\begin{align}\label{Z-n-0-asy}
Z_n(\lambda_0) & = Z_n^{GUE} \exp \left( \frac{2n^2t_1}{\lambda_0+\sqrt{\lambda_0^2-1}}\right)\left(1+O\left( \frac{1}{n^{1/3} \sqrt{\lambda_0-1}}\right)\right),
\nonumber
\\
& = Z_n^{GUE} \exp\left(2n^2t_1(1-\sqrt{2(\lambda_0-1)})\right)(1+O(n^{-1/15})),
%\label{eq:Zn0asy}
\end{align}
where we have made use of the fact that $2n^2t_1=O(n^{1/3})$; see \eqref{eq:tkasy}. For the integral in \eqref{Z-n-asy}, we observe from \eqref{eq:lambdascalingparti} and \eqref{eq:fnear1} that
$$n^{2/3}f(\lambda)=s+O(n^{-2/3}).$$ Thus,
\begin{align}\label{eq:Integral}
&\int_{n^{2/3}f(\lambda)}^{n^{2/3}f(\lambda_0)} \left(a_1(t)-\frac{t^2}{4} \right) dt
\nonumber
\\
&= \int^{n^{2/3}f(\lambda_0)}_{s} \left(a_1(t)-\frac{t^2}{4} \right) dt +O(n^{-2/3})
\nonumber\\
&=\int_{s}^{n^{2/3}f(\lambda_0)} \left(a_1(t)-\frac{t^2}{4} -\frac{\tau_1}{\sqrt{t-s}}\right) dt +2\tau_1\sqrt{n^{2/3}f(\lambda_0)-s}+O(n^{-2/3}).
\end{align}
Note that
 \begin{equation}\label{eq:elem}
2n^2t_1(1-\sqrt{2(\lambda_0-1)}+2\tau_1\sqrt{n^{2/3}f(\lambda_0)-s}=2n^{1/3}\tau_1+o(1).
\end{equation}
This, together with the asymptotic behavior of $a_1$ given in \eqref{eq:a 1 asy}, implies that
\begin{multline}\label{eq:Integral-total}
\int_{s}^{n^{2/3}f(\lambda_0)} \left(a_1(t)-\frac{t^2}{4} -\frac{\tau_1}{\sqrt{t-s}}\right) dt
\\
=\int_{s}^{\infty} \left(a_1(t)-\frac{t^2}{4} -\frac{\tau_1}{\sqrt{t-s}}\right) dt+O\left(\frac{1}{\sqrt{n^{2/3}f(\lambda_0)}}\right),
\end{multline}
where the error bound  $O(1/\sqrt{n^{2/3}f(\lambda_0)}=O(n^{-2/15})=o(1)$.  On account of the relation \eqref{a-1-b-1} between $a_1$ and $b_1$, we obtain from the integration by parts that
\begin{equation}\label{eq:Integration by part}
\int_{s}^{\infty} \left(a_1(t)-\frac{t^2}{4} -\frac{\tau_1}{\sqrt{t-s}}\right) dt= -\int_{s}^{\infty} \left(b_1(t)(t-s) +\frac{\tau_1}{2\sqrt{t-s}}\right) dt,
\end{equation}
Finally, substituting \eqref{Z-n-0-asy}, \eqref{eq:Integral-total} and \eqref{eq:Integration by part} into \eqref{Z-n-asy} gives us \eqref{eq:zn-asy}.

This completes the proof of Theorem \ref{thm:partition asymptotics}.
\qed

% ------------------------------------------------------------------------

\subsection*{Acknowledgment}
The work of Dan Dai was partially supported by grants from the Research Grants Council of the Hong Kong Special Administrative Region, China (Project No. CityU 11300115, CityU 11303016), and by grants from City University of Hong Kong (Project No. 7004864, 7005032). The work of Shuai-Xia Xu was partially supported by National Natural Science Foundation of China under grant number 11571376, GuangDong Natural Science Foundation under grant number 2014A030313176.  The work of Lun Zhang was partially supported by National Natural Science Foundation of China under grant numbers 11822104 and 11501120, by The Program for Professor of Special Appointment (Eastern Scholar) at Shanghai Institutions of Higher Learning, and by Grant EZH1411513 from Fudan University.

\end{document}